\documentclass[acmsmall]{acmart}
\usepackage{graphicx} % Required for inserting images
\usepackage{amsthm,amsmath}
\usepackage{subcaption}

\newtheorem{theorem}{Theorem}

\newtheorem{corollary}{Corollary}
\newcounter{remarkcount}

\usepackage{xcolor}
\usepackage{tikz}
\usetikzlibrary{arrows.meta,calc,positioning}
\usetikzlibrary{decorations.pathmorphing}

\title{Colored Markov Modulated Fluid Queues}
\author{Benny Van Houdt}
\date{}

\begin{document}

\begin{abstract}
   Markov-modulated fluid queues (MMFQs) are a powerful modeling framework for analyzing the performance of computer and communication systems. Their distinguishing feature is that the underlying Markov process evolves on a continuous state space, making them well suited to capture the dynamics of workloads, energy levels, and other performance-related quantities. Although classical MMFQs do not permit jumps in the fluid level, they can still be applied to analyze a wide range of jump processes.
   
In this paper, we generalize the MMFQ framework in a new direction by introducing {\bf colored MMFQs} and 
{\bf colored MMFQs with fluid jumps}. This enriched framework provides an additional form of memory: the color of incoming fluid can be used to keep track of the fluid level when certain events took place. This capability greatly enhances modeling flexibility and enables the analysis of queueing systems that would otherwise be intractable due to the curse of dimensionality or state-space explosion.
\end{abstract}

\maketitle

\section{Introduction}

Markov chain modeling is a powerful technique for evaluating the performance of computer systems and networks \cite{Asmussen2021AppliedProbability,Bolch2006Queueing,HarcholBalter2013Performance,leboudec2010performance,Stewart2009Probability,BruneelKim1993DiscreteTimeATM}, and has also been applied in areas such as risk management and insurance, manufacturing, supply chains, and inventory management.
A Markovian model requires a state descriptor $\Omega$, where the system’s state $s(t)\in\Omega$ at time $t$ determines its future evolution independently of the past. For stable systems, the underlying Markov chain admits a unique stationary distribution, which (under mild conditions) yields performance measures.

When $\Omega$ is finite and of moderate size, the stationary distribution can be obtained by solving a system of linear equations. However, many practical models have very large or infinite state spaces, rendering direct methods such as Gaussian elimination infeasible. This motivated the development of matrix analytic methods, which exploit structural properties of the equations. Since the pioneering work of Neuts in the 1980s \cite{neuts2,neuts1}, these methods have been extensively studied and are covered in several textbooks \cite{latouche1,bini2,Alfa2010Queueing,He2013Fundamentals}. 

Most of the work has focused on Markov chains with denumerable state spaces, such as Quasi-Birth-Death, M/G/1-type Markov chains,
and their tree-structured extensions \cite{takine1,yeung1}. 
A notable exception is the work on Markov-modulated fluid queues \cite{soares3,AhnRamaswami2003,AkarSohraby2004} and their
multi-layered extensions \cite{KankayaAkar2008,BeanOReilly2008}, which have found applications in various areas \cite{GatzianasGeorgiadisTassiulas2010,Bierbooms2012,YaziciAkar2013,Samuelson2017,YaziciAkar2017,TangTan2017,HeZhangYe2018,Simon2020,NoblesseSonenbergBouteLambrechtVanHoudt2022,vanhoudt_nudge,BreuerBadescuSoaresLatoucheRemicheStanford2005,HorvathVanHoudt2012}.
MMFQs can also be leveraged to study certain fluid queue models with jumps \cite{dzial1,vanHoudt29}.

In this paper, we extend the framework of Markov-modulated fluid queues (MMFQs) in a new direction by introducing {\bf colored MMFQs and colored MMFQs with fluid jumps}. In a classical MMFQ, the fluid in the queue can be regarded as having a single color. In contrast, a colored MMFQ allows different colors of fluid to be added to the queue.
The key advantage of introducing multiple colors is that they provide an additional form of memory: the color of incoming fluid can be used to keep track of the fluid level when certain events took place. This enriched modeling capability enables the analysis of queueing systems that would otherwise be intractable using existing approaches, due to the curse of dimensionality or state-space explosion.

For example, when the fluid represents the workload of jobs in a queue, one could assign a different color to the work contributed by each job. This makes it possible to employ a workload-process approach even when the system accommodates only a finite number of jobs. Such a refinement is not achievable with a classical MMFQ, where only the total workload is tracked, without distinguishing contributions from individual jobs. While finite-capacity MMFQs have been studied before, in those models the capacity constraint applies to the total workload in the queue, rather than the number of jobs.

The remainder of the paper is organized as follows. Section \ref{sec:classic} reviews the classic Markov-modulated fluid queue (MMFQ). Sections \ref{sec:coloredMMFQ} and \ref{sec:stationary} introduce the framework of colored MMFQs, with special cases presented in Section \ref{sec:special}. Extensions to colored MMFQs with fluid jumps are developed in Section \ref{sec:jumps}, and applications are discussed in Section \ref{sec:apps}. 
Numerical and runtime results are provided in Section \ref{sec:num}. The proof of Theorem \ref{th:main} is outlined in Section \ref{sec:proof}, and conclusions along with directions for future work are given in Section \ref{sec:future}.

\section{The classic Markov-modulated fluid queue:}\label{sec:classic}

The classic Markov-modulated fluid queue is a Markov process on the state space $\Omega$ of the form 
\[ \Omega = \{(0,i)| i \in S_-) \cup \{(x,i) | x \in  \mathbb{R}^+,i \in S \},\]
where $S = S_- \cup S_+$ is a finite set of states of the so-called background process.
When the state equals $(x,i)$ at time $t$, we refer to $x$ as the fluid level and to
$i$ as the background state. Unless a background transition occurs, an MMFQ evolves as follows\footnote{The generalization to MMFQs with rates that depend on the state $i \in S$ is not hard, see \cite[Section 1.7]{soares3}.}:
\begin{itemize}
    \item when $x > 0$ and $i \in S_-$ the fluid level decreases at
rate $1$,
    \item when $x > 0$ and $i \in S_+$ the fluid level increases at
rate $1$,
\item when $x=0$, meaning $i \in S_-$, the state remains the same.
\end{itemize}
Background transitions can occur from state $(x,i)$ to $(x,j)$ with $x>0$ and are characterized by the matrices:
\begin{itemize}
    \item $(T_{--})_{ij}$: rate of moving from state $(x,i)$ to $(x,j)$  with $i,j \in S_-$,
\item $(T_{-+})_{ij}$: rate of moving from state $(x,i)$ to $(x,j)$  with $i \in S_-, j \in S_+$,
\item $(T_{++})_{ij}$ : rate of moving from state $(x,i)$ to $(x,j)$ with $i,j \in S_+$,
\item $(T_{+-})_{ij}$: rate of moving from state $(x,i)$ to $(x,j)$ with $i \in S_+, j \in S_-$.
\end{itemize}
Background transitions  can also occur in the boundary states $(0,i)$ and are characterized by
\begin{itemize}
    \item $(T_{--}^{(0)})_{ij}$: rate of moving from state $(0,i)$ to $(0,j)$ with $i,j \in S_-$,
    \item $(T_{-+}^{(0)})_{ij}$: rate of moving from state $(0,i)$ to $(0^+,j)$ with $i \in S_-, j\in S_+$,
    which corresponds to a background transition that initiates an increase in the fluid level away from zero.
\end{itemize}
Let 
\[ T=\begin{bmatrix}
    T_{++} & T_{+-} \\ T_{-+} & T_{--}
\end{bmatrix},\]
and define $\xi$ as the invariant vector of $T$ partitioned as $\xi = (\xi_+, \xi_-)$ in the obvious manner. 
An MMFQ has a unique stationary distribution if and only if $\xi_- e > \xi_+ e$. 
The following matrix analytic
method can be used to compute the stationary distribution \cite{soares3}: 
\begin{itemize}
    \item Compute  the element-wise minimal nonnegative solution of the nonsymmetric algebraic Riccati equation (NARE):
\begin{align}\label{eq:Psiclassic} 
0 = T_{++}\Psi + \Psi T_{-+}\Psi + \Psi T_{--} + T_{+-}, 
\end{align}
which is a stochastic matrix $\Psi$ (when the MMFQ has a stationary
distribution). Various numerically stable algorithms like SDA and ADDA can be used to compute $\Psi$ \cite{guo2,guo3,wangADDA}.
Entry $\Psi_{ij}$ holds the first passage probability that, for any $x > 0$, the background state equals $j \in S_-$ when the MMFQ
first returns to fluid level $x$ provided that it started in state $(x,i)$ with $i \in S_+$ at time $0$.
\item The stationary density vector $\pi(x)=[\pi_+(x),\pi_-(x)]$ of the MMFQ is then given by
\begin{align}
    [\pi_+(x),\pi_-(x)] = p_- T_{-+}^{(0)}e^{Kx}  [I,\Psi], 
\end{align}
for $x > 0$, with $K=T_{++} + \Psi T_{-+}$. The vector $p_-$ is the invariant vector
of the generator matrix $T_{--}^{(0)} + T_{-+}^{(0)} \Psi$ (which corresponds to censoring the Markov process
on the states $\{(0,i)|i\in S_-\}$), normalized such that
we get a distribution.
\end{itemize}
An illustration of a sample path of a classic MMFQ with $S=\{1,2,3,4\}$ is shown in Figure \ref{fig:MMFQevo}.
At time $1.2$ the background process moves from state $1$ to $2$ (both in $S_+$) and the fluid continues to increase, while
at time $2.1$ a move to state $4 \in S_-$ occurs and the fluid level starts to decrease.

\begin{figure}
    \centering
\begin{tikzpicture}[x=1.2cm,y=1cm,>=Latex,
  axis/.style={thin},
  inc/.style={very thick},             % slope +1
  dec/.style={very thick,dashed},      % slope -1
  lab/.style={font=\small}]
% --- editable parameters ---
\def\xstart{0.5} % initial level
\def\dA{1.2}  % state 1  ( +1 )
\def\dB{0.9}  % state 2  ( +1 )
\def\dC{1.4}  % state 3  ( -1 )
\def\dD{0.7}  % state 2  ( +1 )
\def\dE{1.8}  % state 1  ( +1 )
\def\dF{1.2}  % state 3  ( -1 )
\def\dG{1.3}  % state 3  ( -1 )

% --- starting point ---
\coordinate (P0) at (0,\xstart);

% --- trajectory segments ---
\draw[inc] (P0) -- node[above,sloped,lab]{+1} ++(\dA,\dA) coordinate (P1);
\draw[inc] (P1) -- node[above,sloped,lab]{+1} ++(\dB,\dB) coordinate (P2);
\draw[dec] (P2) -- node[above,sloped,lab]{-1} ++(\dC,-\dC) coordinate (P3);
\draw[inc] (P3) -- node[above,sloped,lab]{+1} ++(\dD,\dD) coordinate (P4);
\draw[inc] (P4) -- node[above,sloped,lab]{+1} ++(\dE,\dE) coordinate (P5);
\draw[dec] (P5) -- node[above,sloped,lab]{-1} ++(\dF,-\dF) coordinate (P6);
\draw[dec] (P6) -- node[above,sloped,lab]{-1} ++(\dG,-\dG) coordinate (P7);

\path[fill=black!10]
  (P0 |- 0,0) -- (P0) -- (P1) -- (P2) -- (P3) -- (P4) -- (P5) -- (P6) -- (P7) --  (P7 |- 0,0) -- cycle;

% --- axes (extend to last point) ---
\draw[axis,->] (0,0) -- ($(P7 |- 0,0) + (0.6,0)$) node[below right] {$t$};
\draw[axis,->] (0,0) -- (0,5) node[left] {};

% initial marker & label
%\fill (P0) circle[radius=1.2pt];
%\node[lab,anchor=west] at ($(P0)-(0.06,0.18)$) {$x(0)=\xstart$};

% --- light guides at switch times (optional, remove if not needed) ---
\foreach \Q in {P1,P2,P3,P4,P5,P6}
  \draw[axis,dashed] (\Q |- 0,0) -- (\Q |- 0,4.9);

% --- state labels at y=4 between guides ---
\node[lab] at ($(P0 |- 0,4)!0.5!(P1 |- 0,4)$) {1};
\node[lab] at ($(P1 |- 0,4)!0.5!(P2 |- 0,4)$) {2};
\node[lab] at ($(P2 |- 0,4)!0.5!(P3 |- 0,4)$) {4};
\node[lab] at ($(P3 |- 0,4)!0.5!(P4 |- 0,4)$) {2};
\node[lab] at ($(P4 |- 0,4)!0.5!(P5 |- 0,4)$) {1};
\node[lab] at ($(P5 |- 0,4)!0.5!(P6 |- 0,4)$) {3};
\node[lab] at ($(P6 |- 0,4)!0.5!(P7 |- 0,4)$) {4};

\node[lab, anchor=west, text width=2cm] at ($(P7 |- 0,4) + (0.3,0)$) {background state};
\node[lab, anchor=west] at ($(P7) + (0.3,0)$) {fluid level};

\node[lab, anchor=north] at ($(P0 |- 0,0)$) {0};
\node[lab, anchor=north] at ($(P1 |- 0,0)$) {\dA};
\node[lab, anchor=north] at ($(P2 |- 0,0)$) {$\pgfmathparse{\dA+\dB}\pgfmathprintnumber{\pgfmathresult}$};
\node[lab, anchor=north] at ($(P3 |- 0,0)$) {$\pgfmathparse{\dA+\dB+\dC}\pgfmathprintnumber{\pgfmathresult}$};
\node[lab, anchor=north] at ($(P4 |- 0,0)$) {$\pgfmathparse{\dA+\dB+\dC+\dD}\pgfmathprintnumber{\pgfmathresult}$};
\node[lab, anchor=north] at ($(P5 |- 0,0)$) {$\pgfmathparse{\dA+\dB+\dC+\dD+\dE}\pgfmathprintnumber{\pgfmathresult}$};
\node[lab, anchor=north] at ($(P6 |- 0,0)$) {$\pgfmathparse{\dA+\dB+\dC+\dD+\dE+\dF}\pgfmathprintnumber{\pgfmathresult}$};
%\node[lab, anchor=north] at ($(P7 |- 0,0)$) {$\pgfmathparse{\dA+\dB+\dC+\dD+\dE+\dF+\dG}\pgfmathprintnumber{\pgfmathresult}$};
\end{tikzpicture}
    \caption{Illustration of a sample path of an MMFQ with $S_+=\{1,2\}$ and $S_-=\{3,4\}$, background transitions occur
    at times $1.2, 2.1, 3.5, 4.2, 6$ and $7.2$.}
    \label{fig:MMFQevo}
\end{figure}
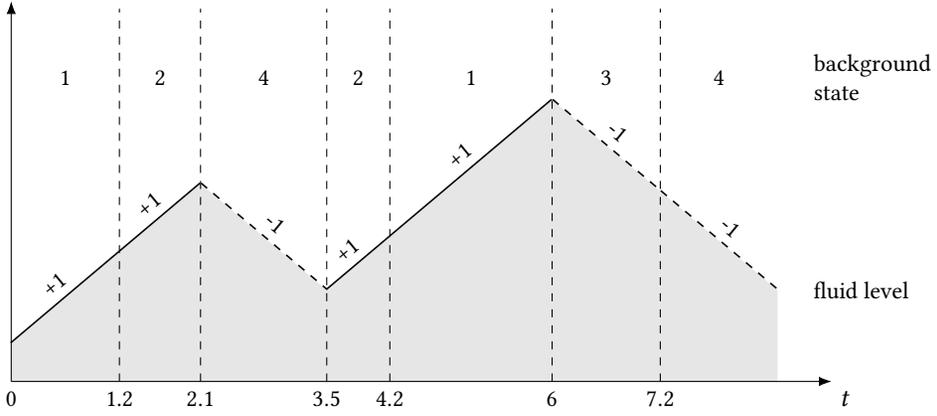

\section{The colored Markov-modulated fluid queue}\label{sec:coloredMMFQ}

\paragraph{The state space}
We now introduce the colored MMFQ with $C$ colors that reduces to the classic MMFQ when $C=1$.
It is a Markov process defined on the state space $\Omega$ of the form 
\[ \Omega = \{(0,i)| i \in S_-\} \cup \left( \bigcup_{c=1}^C \Omega_c\right)\]
with \[ \Omega_c = 
\left\{(x_1,x_2,\ldots,x_C,i) | x_1,\ldots,x_c \geq 0, x_c > 0,
x_{c+1}=\ldots=x_C = 0, i \in S_- \cup S_+^{(c)} \right\},\]
where $S = S_- \cup (\cup_{c=1}^C S^{(c)}_+)$ is the finite set of states of the background process.
The idea of state $(x_1,x_2,\ldots,x_C,i) \in \Omega_c$ is as follows: the total amount of fluid in the fluid
queue equals $\sum_{c=1}^C x_c$ and this amount consists of $x_1$ fluid of color $1$, followed by $x_2$ fluid
of color $2$, etc. The color of the fluid at the top is $c$ when the state is in $\Omega_c$ (as $x_{c+1},\ldots,
x_C$ are zero). Notice that the fluid in the fluid queue is always {\bf ordered by color}, with color $1$
at the bottom (if $x_1 > 0$) and color $c$ at the top for a state in $\Omega_c$. Furthermore, some
of the values in $\{x_1,\ldots,x_{c-1}\}$ can be zero for states in $\Omega_c$. It should also be feasible
to construct colored MMFQs where the colors are not necessarily ordered, but this poses many additional technical challenges
beyond the scope of this paper.

For the background state $i$, we see that a common set of states $S_-$ is used for all colors,
while a color dependent set $S_+^{(c)}$ of states is used when the fluid increases. 
One of the features of a colored MMFQ that makes it a true generalization of a classic MMFQ 
lies in the fact that {\bf the rate matrix of the background process when the fluid decreases depends on the color of the fluid on top}, that is, instead of having a single $T_{--}$ matrix, a colored MMFQ relies on $C$ matrices $T_{--}^{(1)}$ to $T_{--}^{(C)}$. 
In Appendix \ref{app:special2} we show that if the color of the fluid on top is ignored while the fluid goes down,
then a colored MMFQ reduces to a classic MMFQ.
It is possible to generalize the results in this paper such that the set of background states is also
color dependent when the fluid decreases. In this case some additional transition probability matrices are needed that
specify the change in the state of the background process when a color boundary is crossed when the fluid decreases.
This however should not pose any problems.

\paragraph{The evolution of a colored MMFQ}
When the state $s=(x_1,\ldots,x_C,i) \in \Omega_c$ at time $t$, we refer to $x_1+\ldots+x_c$ as the fluid level and to
$i$ as the background state. Unless a background transition occurs, a colored MMFQ evolves as follows:
\begin{itemize}
    \item when $s \in \cup_{c=1}^C \Omega_c$ and $i \in S_-$ the fluid level decreases at
rate $1$,
    \item when $s \in \Omega_c$ and $i \in S^{(c)}_+$ the fluid level increases at
rate $1$,
\item when $s=(0,i)$ with $i \in S_-$, the state remains the same.
\end{itemize}
When the state $s=(x_1,\ldots,x_C,i) \in \Omega_c$ with $i \in S_-$ and the value of $x_c$ hits zero,
then the state of the colored MMFQ becomes part of the set $\Omega_{c'}$ with $c' = \max \{k | x_k > 0, k < c\}$ or
becomes equal to $(0,i)$ if $x_1=\ldots=x_{c-1}=0$. In the first case the color at the top becomes $c'$, in the second
case the fluid level becomes $0$.

We define two kinds of background transitions in a colored MMFQ: {\bf Transitions of the first kind} are within the set $\Omega_c$, for some $c \in \{1,\ldots,C\}$.
These background transitions change the state $s$ from $s=(x_1,\ldots,x_C,i) \in \Omega_c$ to 
state $s'=(x_1,\ldots,x_C,j) \in \Omega_c$ and are characterized by the matrices:
\begin{itemize}
    \item $(T^{(c)}_{--})_{ij}$: rate of moving from state $s$ to $s'$  with $i,j \in S_-$,
\item $(T^{(c)}_{-+})_{ij}$: rate of moving from state $s$ to $s'$  with $i \in S_-, j \in S_+^{(c)}$,
\item $(T^{(c)}_{++})_{ij}$ : rate of moving from state $s$ to $s'$ with $i,j \in S_+^{(c)}$,
\item $(T^{(c)}_{+-})_{ij}$: rate of moving from state $s$ to $s'$ with $i \in S_+^{(c)}, j \in S_-$.
\end{itemize}
It is important to note here that all these rates depend on the color $c$ of the fluid at the top, 
that is, they depend on $c$ when $s \in \Omega_c$, even when the
fluid decreases. This means that if the color at the top changes during an interval where the fluid decreases
then so does the rate matrix associated to the above mentioned background transitions.

{\bf Background transitions of the second kind} are the ones that correspond to a transition 
from a state $s=(x_1,\ldots,x_C,i) \in \Omega_c$
out of the set $\Omega_c$. These are characterized by
\begin{itemize}
    \item $(T_{-+}^{(c,c')})_{ij}$ with $c' > c$: rate of moving from state $s$ to 
    $s'=(x_1,\ldots,x_c,0,\ldots,0,0^+,0,\ldots,0,j)$ with $i \in S_-, j\in S_+^{(c')}$,
    where the $0^+$ occurs in position $c'$. Such a background transition initiates an increase in the fluid level 
    using color $c' > c$.
    \item $(T_{++}^{(c,c')})_{ij}$ with $c' > c$: rate of moving from state $s$ to 
    $s'=(x_1,\ldots,x_c,0,\ldots,0,0^+,0,\ldots,0,j)$ with $i \in S_+^{(c)}, j\in S_+^{(c')}$.
    In this case the fluid continues to increase, but the color being added changes from $c$ to $c'$.
\end{itemize}
We demand that $c' > c$ such that the colors of the fluid remain ordered.

Background transitions of the first and second kind can also occur in the boundary states $(0,i)$ and are characterized by
\begin{itemize}
    \item $(T_{--}^{(0)})_{ij}$: rate of moving from state $(0,i)$ to $(0,j)$ with $i,j \in S_-$,
    \item $(T_{-+}^{(0,c)})_{ij}$: rate of moving from state $(0,i)$ to $(0,\ldots,0,0^+,0\ldots,0,j)$ with 
    $i \in S_-, j\in S_+^{(c)}$,
    which corresponds to a background transition that initiates an increase in the fluid level away from zero using color $c$.
\end{itemize}

An illustration of a sample path of a colored MMFQ with $2$ colors is presented in Figure \ref{fig:MMFQevo2}.
The set $S=\{1,2,3,4,5,6\}$ with $S_+^{(1)}=\{1\}, S_+^{(2)}=\{2,3,4\}$ and $S_-=\{5,6\}$.
A background transition of the second kind occurs at time $1.2$, while a background transition of the first kind takes place at time $2.1$.
At time $3$ the fluid at the top changes color, meaning the rate matrix of the background process changes
from $T_{--}^{(2)}$ to $T_{--}^{(1)}$.

\begin{figure}
    \centering
\begin{tikzpicture}[x=1.2cm,y=1cm,>=Latex,
  axis/.style={thin},
  inc/.style={very thick},             % slope +1
  dec/.style={very thick,dashed},      % slope -1
  lab/.style={font=\small}]
% --- editable parameters ---
\def\xstart{0.5} % initial level
\def\dA{1.2}  % state 1  ( +1 )
\def\dB{0.9}  % state 2  ( +1 )
\def\dC{1.5}  % state 3  ( -1 )
\def\dD{1.1}  % state 2  ( +1 )
\def\dE{1}  % state 1  ( +1 )
\def\dF{0.7}  % state 3  ( -1 )
\def\dG{0.6}  % state 3  ( -1 )
\def\dH{1.8}  % state 1  ( +1 )

% --- starting point ---
\coordinate (P0) at (0,\xstart);

% --- trajectory segments ---
\draw[inc] (P0) -- node[above,sloped,lab]{+1} ++(\dA,\dA) coordinate (P1);
\draw[inc,red] (P1) -- node[above,sloped,lab]{+1} ++(\dB,\dB) coordinate (P2);
\draw[dec,red] (P2) -- node[above,sloped,lab]{-1} ++(\dB,-\dB) coordinate (P3a);
\draw[dec] (P3a) -- node[above,sloped,lab]{-1} ++(\dC-\dB,-\dC+\dB) coordinate (P3);
\draw[inc] (P3) -- node[above,sloped,lab]{+1} ++(\dD,\dD) coordinate (P4);
\draw[inc,red] (P4) -- node[above,sloped,lab]{+1} ++(\dE,\dE) coordinate (P5);
\draw[dec,red] (P5) -- node[above,sloped,lab]{-1} ++(\dF,-\dF) coordinate (P6);
\draw[inc,red] (P6) -- node[above,sloped,lab]{+1} ++(\dG,\dG) coordinate (P7);
\draw[dec,red] (P7) -- node[above,sloped,lab]{-1} ++(0.9,-0.9) coordinate (P7a);
\draw[dec] (P7a) -- node[above,sloped,lab]{-1} ++(\dH-0.9,-\dH+0.9) coordinate (P8);

\path[fill=black!10]
  (P0 |- 0,0) -- (P0) -- (P1) -- (P3a) -- (P3) -- (P4) -- (P7a) -- (P8) -- (P8 |- 0,0) -- cycle;

\path[fill=red!10]
 (P1) -- (P2) -- (P3a) -- cycle;

\path[fill=red!10]
 (P4) -- (P5) -- (P6) -- (P7) -- (P7a) -- cycle;

\draw[dotted] (P1) -- (P3a);
\draw[dotted] (P3a) -- ($(P3a |- 0,0)$);
\draw[dotted] (P4) -- (P7a);
\draw[dotted] (P7a) -- ($(P7a |- 0,0)$);

% --- axes (extend to last point) ---
\draw[axis,->] (0,0) -- ($(P8 |- 0,0) + (0.6,0)$) node[below right] {$t$};
\draw[axis,->] (0,0) -- (0,5) node[left] {};

% initial marker & label
%\fill (P0) circle[radius=1.2pt];
%\node[lab,anchor=west] at ($(P0)-(0.06,0.18)$) {$x(0)=\xstart$};

% --- light guides at switch times (optional, remove if not needed) ---
\foreach \Q in {P1,P2,P3,P4,P5,P6,P7}
  \draw[axis,dashed] (\Q |- 0,0) -- (\Q |- 0,4.9);

% --- state labels at y=4 between guides ---
\node[lab] at ($(P0 |- 0,4)!0.5!(P1 |- 0,4)$) {1};
\node[lab] at ($(P1 |- 0,4)!0.5!(P2 |- 0,4)$) {3};
\node[lab] at ($(P2 |- 0,4)!0.5!(P3 |- 0,4)$) {6};
\node[lab] at ($(P3 |- 0,4)!0.5!(P4 |- 0,4)$) {1};
\node[lab] at ($(P4 |- 0,4)!0.5!(P5 |- 0,4)$) {4};
\node[lab] at ($(P5 |- 0,4)!0.5!(P6 |- 0,4)$) {5};
\node[lab] at ($(P6 |- 0,4)!0.5!(P7 |- 0,4)$) {2};
\node[lab] at ($(P7 |- 0,4)!0.5!(P8 |- 0,4)$) {5};

\node[lab, anchor=west,text width=2cm] at ($(P8 |- 0,4) + (0.3,0)$) {background state};
\node[lab, anchor=west] at ($(P8) + (0.3,0)$) {fluid level};

\node[lab, anchor=north] at ($(P0 |- 0,0)$) {0};
\node[lab, anchor=north] at ($(P1 |- 0,0)$) {\dA};
\node[lab, anchor=north] at ($(P2 |- 0,0)$) {$\pgfmathparse{\dA+\dB}\pgfmathprintnumber{\pgfmathresult}$};
\node[lab, anchor=north] at ($(P3 |- 0,0)$) {$\pgfmathparse{\dA+\dB+\dC}\pgfmathprintnumber{\pgfmathresult}$};
\node[lab, red, anchor=north] at ($(P3a |- 0,0)$) {$\pgfmathparse{\dA+\dB+\dB}\pgfmathprintnumber{\pgfmathresult}$};
\node[lab, anchor=north] at ($(P4 |- 0,0)$) {$\pgfmathparse{\dA+\dB+\dC+\dD}\pgfmathprintnumber{\pgfmathresult}$};
\node[lab, anchor=north] at ($(P5 |- 0,0)$) {$\pgfmathparse{\dA+\dB+\dC+\dD+\dE}\pgfmathprintnumber{\pgfmathresult}$};
\node[lab, anchor=north] at ($(P6 |- 0,0)$) {$\pgfmathparse{\dA+\dB+\dC+\dD+\dE+\dF}\pgfmathprintnumber{\pgfmathresult}$};
\node[lab, anchor=north] at ($(P7 |- 0,0)$) {$\pgfmathparse{\dA+\dB+\dC+\dD+\dE+\dF+\dG}\pgfmathprintnumber{\pgfmathresult}$};
\node[lab, red, anchor=north] at ($(P7a |- 0,0)$) {$\pgfmathparse{\dA+\dB+\dC+\dD+\dE+\dF+\dG+0.9}\pgfmathprintnumber{\pgfmathresult}$};
%\node[lab, anchor=north] at ($(P8 |- 0,0)$) {$\pgfmathparse{\dA+\dB+\dC+\dD+\dE+\dF+\dG+\dH}\pgfmathprintnumber{\pgfmathresult}$};

\end{tikzpicture}
    \caption{Illustration of a sample path of a $2$-colored MMFQ with $S^{(1)}_+=\{1\}$, $S^{(2)}_+=\{2,3,4\}$ and $S_-=\{5,6\}$, background transitions occur
    at times $1.2, 2.1, 3.6, 4.7, 5.7, 6.4$ and $7$. There is no background transition at times $3$ and $7.9$, but the rate matrix of the background
    process does change from $T_{--}^{(2)}$ to $T_{--}^{(1)}$.}
    \label{fig:MMFQevo2}
\end{figure}
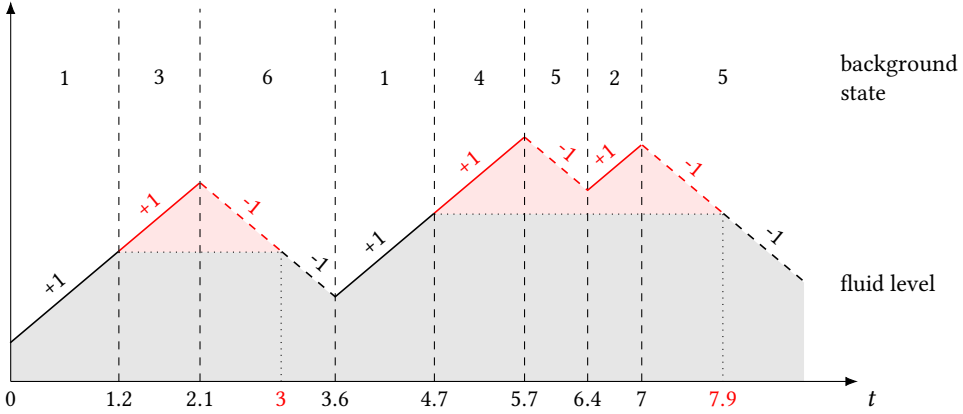

\section{Stationary distribution of a colored Markov-modulated fluid queue}\label{sec:stationary}

In this section we present our main result, the proof of which is deferred to Section \ref{sec:proof}.
We first introduce a set of $C$ matrices $\Psi_1, \ldots,\Psi_C$. These matrices hold the first passage
probabilities for the $C$ colors of the colored MMFQ. More specifically, entry $(\Psi_c)_{ij}$ holds the first passage probability that, for any $\vec x = (x_1,\ldots,x_c,0,\ldots,0) \in \mathbb{R}^C$ with $x_c > 0$, the background state equals $j \in S_-$ when the MMFQ
first returns to fluid level $x=\sum_{c=1}^C x_c$ provided that it started in state $s=(\vec x,i) \in \Omega_c$ with 
$i \in S^{(c)}_+$ at time $0$.
Note that when the fluid level returns to level $x$ via some state $s'=(\vec x',j)$, we must have $\vec x' = \vec x$. 
If the colored MMFQ is positive (Harris) recurrent, then the matrices  $\Psi_1,\ldots,\Psi_C$ are stochastic matrices. 

We can compute the $\Psi_c$ matrices, for $c=1,\ldots,C$, using a backward recursion by first solving
\begin{align}\label{eq:NAREC}
  T_{+-}^{(C)} + T_{++}^{(C)}&\Psi_C
    +\Psi_CT_{--}^{(C)}+
    \Psi_C T_{-+}^{(C)}\Psi_C  = 0.
\end{align}
and subsequently solving
\begin{align}\label{eq:NAREs}
    T_{++}^{(c)}\Psi_c + & \left( T_{+-}^{(c)} + \sum_{\ell > c} T_{++}^{(c,\ell)}\Psi_\ell \right)
    +\Psi_cT_{-+}^{(c)}\Psi_c +
   \Psi_c \left(T_{--}^{(c)}+\sum_{\ell > k} T_{-+}^{(c,\ell)}\Psi_\ell \right) = 0,
\end{align}
for $c=C-1,\ldots,1$. The equation for $\Psi_C$ is identical to \eqref{eq:Psiclassic}, except that a superscript $^{(C)}$ is added to the
rate matrices $T_{--},T_{-+},T_{+-}$ and $T_{++}$. This is as expected as the colored MMFQ cannot make any background transitions
of the second kind when the color is $C$ and therefore it evolves in the same manner as a classic MMFQ as long as 
the fluid level remains above $x$ (so the
color on top remains equal to $C$).

Define
\begin{align*}
    \tilde T_{++}^{(c)} =  T_{++}^{(c)}, \ \ \  \ \ \ \  \tilde T_{+-}^{(c)} = T_{+-}^{(c)} + \sum_{\ell > c} T_{++}^{(c,\ell)}\Psi_\ell,  \ \ \  \ \ \ \ 
    \tilde T_{-+}^{(c)} =  T_{-+}^{(c)}, \ \ \  \ \ \ \  \tilde T_{--}^{(c)} = T_{--}^{(c)}+\sum_{\ell > c} T_{-+}^{(c,\ell)}\Psi_\ell,
\end{align*}
then \eqref{eq:NAREs} can be rewritten as 
\begin{align*}
    \tilde T_{++}^{(c)}\Psi_c + & \tilde T_{+-}^{(c)} 
    +\Psi_c \tilde T_{-+}^{(c)}\Psi_c +
   \Psi_c  \tilde T_{--}^{(c)} = 0,
\end{align*}
for $c=C-1,\ldots,1$.
Imagine that the colored MMFQ was in state $s=(x_1,\ldots,x_C,i) \in \Omega_k$ at time $0$. Equation \eqref{eq:NAREs} can now be understood by noting that the $4$ matrices $\tilde T_{--}^{(c)},\tilde T_{-+}^{(c)},\tilde T_{+-}^{(c)}$ and $\tilde T_{++}^{(c)}$ capture the evolution
of the colored MMFQ until the first return to level $x=\sum_{c=1}^C x_c$ if we censor out the time intervals where the fluid color becomes
part of the set $\{c+1,\ldots,C\}$. For instance, when using just $2$ colors a background transition from  state $s=(x_1,0,i)$ with $i \in S_+^{(1)}$ to state $s'=(x_1,0,j)$ with $j \in S_-$ can occur in two manners in the censored process:
\begin{itemize}
    \item At rate $(T_{+-})_{ij}$ the background transition is immediate and no censoring occurs.
    \item A background transition of the second kind occurs to some state $s''=(x_1,0^+,v)$ for some $v \in S^{(2)}_+$
    and the first return to fluid level $x_1$ occurs via state $(x_1,0,j)$. This event occurs at rate 
    $\sum_v (T_{++}^{(1,2)})_{iv} (\Psi_2)_{vj} = (T_{++}^{(1,2)}\Psi_2)_{ij}$.  
\end{itemize}
Similarly a background transition from  state $s=(x_1,0,i)$ with $i \in S_-$ to state $s'=(x_1,0,j)$ with $j \in S_-$ can occur in two manners in the censored process (simply replace $(T_{+-})_{ij}$  by $(T_{--})_{ij}$ and   $T_{++}^{(1,2)}$ by $T_{-+}^{(1,2)}$).
This is illustrated in Figure \ref{fig:MMFQevo3}.

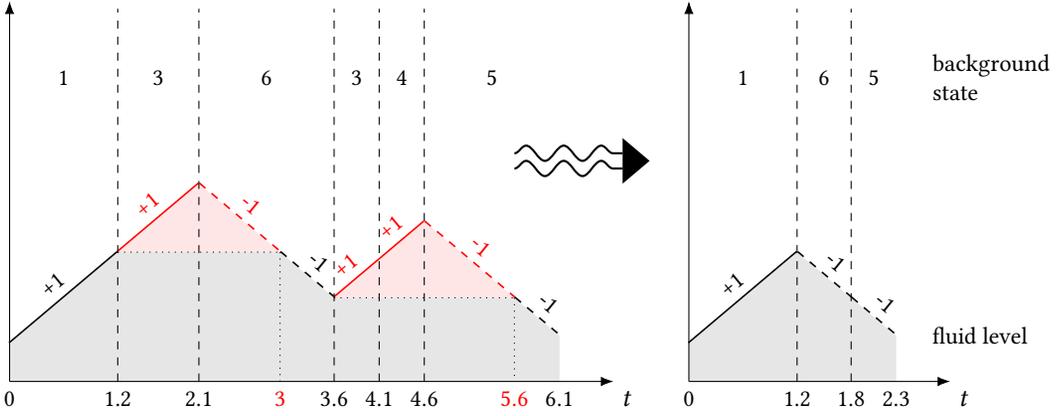
\begin{figure}
\centering
\begin{minipage}{0.6\textwidth}
\begin{tikzpicture}[x=1.2cm,y=1cm,>=Latex,
  axis/.style={thin},
  inc/.style={very thick},             % slope +1
  dec/.style={very thick,dashed},      % slope -1
  lab/.style={font=\small}]
% --- editable parameters ---
\def\xstart{0.5} % initial level
\def\dA{1.2}  % state 1  ( +1 )
\def\dB{0.9}  % state 2  ( +1 )
\def\dC{1.5}  % state 3  ( -1 )
\def\dD{0.5}  % state 2  ( +1 )
\def\dE{0.5}  % state 1  ( +1 )
\def\dF{1}  % state 3  ( -1 )
\def\dG{0.5}  % state 3  ( -1 )

% --- starting point ---
\coordinate (P0) at (0,\xstart);

% --- trajectory segments ---
\draw[inc] (P0) -- node[above,sloped,lab]{+1} ++(\dA,\dA) coordinate (P1);
\draw[inc,red] (P1) -- node[above,sloped,lab]{+1} ++(\dB,\dB) coordinate (P2);
\draw[dec,red] (P2) -- node[above,sloped,lab]{-1} ++(\dB,-\dB) coordinate (P3a);
\draw[dec] (P3a) -- node[above,sloped,lab]{-1} ++(\dC-\dB,-\dC+\dB) coordinate (P3);
\draw[inc,red] (P3) -- node[above,sloped,lab]{+1} ++(\dD,\dD) coordinate (P4);
\draw[inc,red] (P4) -- node[above,sloped,lab]{+1} ++(\dE,\dE) coordinate (P5);
\draw[dec,red] (P5) -- node[above,sloped,lab]{-1} ++(\dF,-\dF) coordinate (P6a);
\draw[dec] (P6a) -- node[above,sloped,lab]{-1} ++(\dG,-\dG) coordinate (P6);

\path[fill=black!10]
  (P0 |- 0,0) -- (P0) -- (P1) -- (P3a) -- (P3) -- (P6a) -- (P6) -- (P6 |- 0,0) -- cycle;

\path[fill=red!10]
 (P1) -- (P2) -- (P3a) -- cycle;

\path[fill=red!10]
 (P3) -- (P5) -- (P6a)  -- cycle;

\draw[dotted] (P1) -- (P3a);
\draw[dotted] (P3a) -- ($(P3a |- 0,0)$);
\draw[dotted] (P3) -- (P6a);
\draw[dotted] (P6a) -- ($(P6a |- 0,0)$);

% --- axes (extend to last point) ---
\draw[axis,->] (0,0) -- ($(P6 |- 0,0) + (0.6,0)$) node[below right] {$t$};
\draw[axis,->] (0,0) -- (0,5) node[left] {};

% initial marker & label
%\fill (P0) circle[radius=1.2pt];
%\node[lab,anchor=west] at ($(P0)-(0.06,0.18)$) {$x(0)=\xstart$};

% --- light guides at switch times (optional, remove if not needed) ---
\foreach \Q in {P1,P2,P3,P4,P5}
  \draw[axis,dashed] (\Q |- 0,0) -- (\Q |- 0,4.9);

% --- state labels at y=4 between guides ---
\node[lab] at ($(P0 |- 0,4)!0.5!(P1 |- 0,4)$) {1};
\node[lab] at ($(P1 |- 0,4)!0.5!(P2 |- 0,4)$) {3};
\node[lab] at ($(P2 |- 0,4)!0.5!(P3 |- 0,4)$) {6};
\node[lab] at ($(P3 |- 0,4)!0.5!(P4 |- 0,4)$) {3};
\node[lab] at ($(P4 |- 0,4)!0.5!(P5 |- 0,4)$) {4};
\node[lab] at ($(P5 |- 0,4)!0.5!(P6 |- 0,4)$) {5};

\tikzset{wobbly/.style={decorate, decoration={snake, amplitude=1mm, segment length=5mm}, thick}}

\draw[wobbly] ($(P6 |- 0,2.8) - (0.5,0)$) -- ($(P6 |- 0,2.8) + (0.7,0)$) node[midway,above]{};
\draw[wobbly] ($(P6 |- 0,3) - (0.5,0)$) -- ($(P6 |- 0,3) + (0.7,0)$) node[midway,above]{};
\path[fill=black]
 ($(P6 |- 0,2.8) + (0.7,-0.2)$) -- ($(P6 |- 0,3) + (0.7,0.2)$) -- ($(P6 |- 0,2.9) + (1,0)$)  -- cycle;

%\node[lab, anchor=west,text width=2cm] at ($(P6 |- 0,4) + (0.3,0)$) {background state};
%\node[lab, anchor=west] at ($(P6) + (0.3,0)$) {fluid level};

\node[lab, anchor=north] at ($(P0 |- 0,0)$) {0};
\node[lab, anchor=north] at ($(P1 |- 0,0)$) {\dA};
\node[lab, anchor=north] at ($(P2 |- 0,0)$) {$\pgfmathparse{\dA+\dB}\pgfmathprintnumber{\pgfmathresult}$};
\node[lab, anchor=north] at ($(P3 |- 0,0)$) {$\pgfmathparse{\dA+\dB+\dC}\pgfmathprintnumber{\pgfmathresult}$};
\node[lab, red, anchor=north] at ($(P3a |- 0,0)$) {$\pgfmathparse{\dA+\dB+\dB}\pgfmathprintnumber{\pgfmathresult}$};
\node[lab, anchor=north] at ($(P4 |- 0,0)$) {$\pgfmathparse{\dA+\dB+\dC+\dD}\pgfmathprintnumber{\pgfmathresult}$};
\node[lab, anchor=north] at ($(P5 |- 0,0)$) {$\pgfmathparse{\dA+\dB+\dC+\dD+\dE}\pgfmathprintnumber{\pgfmathresult}$};
\node[lab, red, anchor=north] at ($(P6a |- 0,0)$) {$\pgfmathparse{\dA+\dB+\dC+\dD+\dE+\dF}\pgfmathprintnumber{\pgfmathresult}$};
\node[lab, anchor=north] at ($(P6 |- 0,0)$) {$\pgfmathparse{\dA+\dB+\dC+\dD+\dE+\dF+\dG}\pgfmathprintnumber{\pgfmathresult}$};
\end{tikzpicture}
\end{minipage}
\hfill
\begin{minipage}{0.35\textwidth}
\begin{tikzpicture}[x=1.2cm,y=1cm,>=Latex,
  axis/.style={thin},
  inc/.style={very thick},             % slope +1
  dec/.style={very thick,dashed},      % slope -1
  lab/.style={font=\small}]
% --- editable parameters ---
\def\xstart{0.5} % initial level
\def\dA{1.2}  % state 1  ( +1 )
\def\dB{0.6}  % state 3  ( -1 )
\def\dC{0.5}  % state 2  ( +1 )

% --- starting point ---
\coordinate (P0) at (0,\xstart);

% --- trajectory segments ---
\draw[inc] (P0) -- node[above,sloped,lab]{+1} ++(\dA,\dA) coordinate (P1);
\draw[dec] (P1) -- node[above,sloped,lab]{-1} ++(\dB,-\dB) coordinate (P2);
\draw[dec] (P2) -- node[above,sloped,lab]{-1} ++(\dC,-\dC) coordinate (P3);

\path[fill=black!10]
  (P0 |- 0,0) -- (P0) -- (P1) -- (P2) -- (P3) -- (P3 |- 0,0) -- cycle;

% --- axes (extend to last point) ---
\draw[axis,->] (0,0) -- ($(P3 |- 0,0) + (0.6,0)$) node[below right] {$t$};
\draw[axis,->] (0,0) -- (0,5) node[left] {};

% initial marker & label
%\fill (P0) circle[radius=1.2pt];
%\node[lab,anchor=west] at ($(P0)-(0.06,0.18)$) {$x(0)=\xstart$};

% --- light guides at switch times (optional, remove if not needed) ---
\foreach \Q in {P1,P2}
  \draw[axis,dashed] (\Q |- 0,0) -- (\Q |- 0,4.9);

% --- state labels at y=4 between guides ---
\node[lab] at ($(P0 |- 0,4)!0.5!(P1 |- 0,4)$) {1};
\node[lab] at ($(P1 |- 0,4)!0.5!(P2 |- 0,4)$) {6};
\node[lab] at ($(P2 |- 0,4)!0.5!(P3 |- 0,4)$) {5};

\node[lab, anchor=west,text width=2cm] at ($(P3 |- 0,4) + (0.3,0)$) {background state};
\node[lab, anchor=west] at ($(P3) + (0.3,0)$) {fluid level};

\node[lab, anchor=north] at ($(P0 |- 0,0)$) {0};
\node[lab, anchor=north] at ($(P1 |- 0,0)$) {\dA};
\node[lab, anchor=north] at ($(P2 |- 0,0)$) {$\pgfmathparse{\dA+\dB}\pgfmathprintnumber{\pgfmathresult}$};
\node[lab, anchor=north] at ($(P3 |- 0,0)$) {$\pgfmathparse{\dA+\dB+\dC}\pgfmathprintnumber{\pgfmathresult}$};
\end{tikzpicture}
\end{minipage}

    \caption{Illustration of a sample path of a $2$-colored MMFQ with $S^{(1)}_+=\{1\}$, $S^{(2)}_+=\{2,3,4\}$ and $S_-=\{5,6\}$.
    Left: original sample path, Right: censored path. Intervals $(1.2,3)$ and $(3.6,5.6)$ of the original sample path
    are censored out.}
    \label{fig:MMFQevo3}
\end{figure}

To state our main theorem we also define the matrices 
\begin{align}\label{eq:Kk}
    K_c = T_{++}^{(c)}+ \Psi_c T_{-+}^{(c)},
\end{align} 
for $c=1,\ldots,C$ and the upper block triangular matrix $K$ of size $\sum_{c=1}^C |S_+^{(c)}|$:
\begin{align}\label{eq:K}
K=T_{++} +\begin{bmatrix}
    \Psi_1  &   \\
     & \ddots &  \\
     &  &  \Psi_C  
\end{bmatrix} T_{-+},
\end{align}
where
\begin{align*}
    T_{++} = \begin{bmatrix}
    T_{++}^{(1)} & T_{++}^{(1,2)} & \ldots & T_{++}^{(1,C)} \\
    0& T_{++}^{(2)} & \ldots & T_{++}^{(2,C)} \\
    \vdots & \ddots & \ddots & \vdots \\
    0 & \ldots & 0 & T_{++}^{(C)} 
\end{bmatrix},  \ \mbox{  and }  \ T_{-+} = \begin{bmatrix}
    T_{-+}^{(1)} & T_{-+}^{(1,2)} & \ldots & T_{-+}^{(1,C)} \\
    0& T_{-+}^{(2)} & \ldots & T_{-+}^{(2,C)} \\
    \vdots & \ddots & \ddots & \vdots \\
    0 & \ldots & 0 & T_{-+}^{(C)}. 
\end{bmatrix}.
\end{align*}

As stated before, positive (Harris) recurrence of the colored MMFQ implies that the $\Psi_k$ matrices are stochastic. 
Positive (Harris) recurrence can be verified as
follows. 
%The matrix $\Psi_C$ is stochastic if $\xi^{(C)}_+ e < \xi^{(C)}_- e$ due to the positive recurrence condition of a classic MMFQ.
For $\Psi_k$, with $k=1,\ldots,C$, define the null vector $\xi^{(c)} = (\xi^{(c)}_+,\xi^{(c)}_-)$ of
\[ T^{(c)} = \begin{bmatrix}
    T_{++}^{(c)} & T_{+-}^{(c)}+\sum_{\ell > k}  T_{++}^{(k,\ell)}\Psi_\ell\\ 
    T_{-+}^{(c)} & T_{--}^{(C)}+\sum_{\ell > k}  T_{-+}^{(k,\ell)}\Psi_\ell  
\end{bmatrix} = \begin{bmatrix}
    \tilde T_{++}^{(c)} & \tilde T_{+-}^{(c)}\\ 
    \tilde T_{-+}^{(c)} & \tilde T_{--}^{(c)}  
\end{bmatrix}.\]
%The matrix $\Psi_k$ is stochastic if  $\xi^{(c)}_+ e < \xi^{(c)}_- e$, for $k=1,\ldots,C-1$.
The colored MMFQ is positive (Harris) recurrent if and only if $\xi^{(c)}_+ e < \xi^{(c)}_- e$, for $c=1,\ldots,C$.
This can be noted as follows. Assume $\xi^{(c)}_+ e \geq \xi^{(c)}_- e$ for some $c$, pick this $c$ as large as possible.
Either $\Psi_c$ is sub-stochastic and the states in $\Omega_c$ are transient, or $\Psi_c$ is stochastic and the mean return time
from a state in $s \in \Omega_c$ to level $x=\sum_i x_i$ is infinite even when not accounting for the time 
when the color at the top of the fluid is in $\{c+1,\ldots,C\}$. Hence, the Markov process is not positive (Harris) recurrent.
If $\xi^{(c)}_+ e < \xi^{(c)}_- e$ does hold for all $c$, then the matrices $K_c$ are all sub-generator matrices and
by Theorem \ref{th:main} the Markov process has a stationary distribution, meaning it is positive (Harris) recurrent.

Let $\vec x = (x_1,\ldots,x_C) \in \mathbb{R}^C$ and let $\pi_-(\vec x)$ be the vector
holding the stationary densities of the states $(\vec x,i)$, for $i \in S_-$. Similarly,  let $\pi_+(\vec x)$ be the vector
holding the stationary densities of the states $(\vec x,i)$, for $i \in S_+^{(c)}$, where $c$ is the largest index such
that $x_c > 0$. Finally, let $p_-$ be the vector holding the stationary  probabilities that the colored MMFQ is
in state $(0,i)$ for $i \in S_-$.

\begin{theorem}\label{th:main}
The stationary densities of the colored fluid queue characterized by the matrices $T_{--}^{(0)}$,$T_{-+}^{(0,c)}$, 
$T_{++}$, $T_{-+}$, $T_{--}^{(c)}$, 
$T_{+-}^{(c)}$, for $c=1,\ldots,C$ can be expressed as follows. Assume $\vec x \in \mathbb{R}^C$ and let
$1 \leq c_1 < c_2 < \ldots < c_n \leq c$ be the subset of indices for which $\vec x$ is positive, then
\begin{align}\label{eq:pix}
    [\pi_+(\vec x),\pi_-(\vec x)] = p_- T_{-+}^{(0,c_1)} \left( \prod_{i=1}^{n-1} e^{K_{c_i}x_{c_i}} 
    \left(T_{++}^{(c_i,c_{i+1})}+\Psi_{c_i} T_{-+}^{(c_i,c_{i+1})}\right) \right) e^{K_{c_n} x_{c_n}} [I,\Psi_{c_n}]. 
\end{align}
The stationary  probability vector $p_-$ solves
\begin{align}\label{eq:pmin}
     p_- \left(T_{--}^{(0)} + \sum_{c=1}^C T_{-+}^{(0,c)} \Psi_c \right)= 0.
\end{align}
This vector is normalized by (where $e$ is a column vector of ones of appropriate size)
\begin{align}\label{eq:norm}
    p_- \left(e +2 [T_{-+}^{(0,1)} \ \ldots \  T_{-+}^{(0,C)}] K^{-1} e \right)= 1.
\end{align} 
\end{theorem}
\begin{proof}
The proof is presented in Section \ref{sec:proof}.    
\end{proof}

The next corollary indicates how to compute some important performance measures such as the distribution of the
fluid level and the distribution of the color on top of the fluid queue. 

\begin{corollary}
Let $\Xi$ be the random variable representing the fluid level defined as $\sum_{c=1}^C x_c$ for state $s=(\vec x,i)$ and
let $\Gamma$ be the random variable representing the color at the top of the fluid given by $\max \{c | x_c > 0\}$ for
state $s=(\vec x,i)$ and by $0$ for state $s=(0,i)$, then
\begin{align}\label{eq:Xi}
    P[\Xi \leq x] &= p_- \left(e +2 [T_{-+}^{(0,1)} \ \ldots \  T_{-+}^{(0,C)}] e^{Kx} e \right), \\
    P[\Gamma = c] &=  2p_- [T_{-+}^{(0,1)} \ \ldots \  T_{-+}^{(0,C)}] (-K)^{-1} w_c, 
    \label{eq:Gamma}
\end{align}
with $P[\Gamma = 0]=p_-e$ and $w_c$ is a size $\sum_{k=1}^{C} |S_+^{(k)}|$ column vector with
\[
(w_c)_i = \left\{ \begin{array}{ll}
     1& \sum_{k=1}^{c-1} |S_+^{(k)}|< i \leq \sum_{k=1}^{c} |S_+^{(k)}|\\
     0& otherwise.
\end{array}\right.
\]
\end{corollary}
\begin{proof}
The expression for   $P[\Xi \leq x]$ follows from \eqref{eq:pix} and \eqref{eq:K} by noting that $[I,\Psi_{c_n}]e=2e$ and
\[ P[\Xi \leq x] = p_-e + \int_{\{\vec x \ | \ \sum_{c=1}^C x_c \leq x\}} [\pi_+(\vec x),\pi_-(\vec x)]e\ d\vec x. \]
The probability $P[\Gamma = c]$ is derived by observing that
\[ P[\Gamma =c] = \int_{\{ \vec x\ | \ x_c > 0, x_{c+1}=\ldots=x_C=0\}} [\pi_+(\vec x),\pi_-(\vec x)]e\ d\vec x. \]
Note that the vector $w_c$ simply selects the columns of $(-K^{-1})$ that correspond to states with color $c$ on top of the fluid.
\end{proof}

\section{Some Special Cases}\label{sec:special}
Before proceeding to the applications we discuss two special cases in this section. 

\subsection{The  matrices $T_{++}^{(c,c')}=0$ and $T_{-+}^{(c,c')}=0$ for $c' > c+1$: no color skipping}\label{sec:special0}
Assume that the matrices $T_{++}^{(c,c')}$ and $T_{-+}^{(c,c')}$ are zero for $c' > c+1$. This means colors cannot be skipped
and having color $c$ on top of the stack implies that there is also some fluid of colors $1$ to $c-1$ present on the stack.
In this case \eqref{eq:pix} simplifies to 
\begin{align*}
    [\pi_+(x_1,\ldots,x_c,\vec 0),\pi_-(x_1,\ldots,x_c,\vec 0)] = 
    p_- T_{-+}^{(0,1)} \left( \prod_{i=1}^{c-1} e^{K_{i}x_i} (T_{++}^{(i,i+1)}+\Psi_{i} T_{-+}^{(i,i+1)}) \right) e^{K_{c} x_{c}} [I,\Psi_{c}], 
\end{align*}
where $x_1,\ldots,x_c > 0$ and $\vec 0$ is a zero vector of length $C-c$.
The distribution of $\Gamma$, computed in general via  \eqref{eq:Gamma}, can now be computed in a {\bf time that is linear in $C$} using
the following recursion:
\begin{itemize}
    \item Set $v_0=p_-, v_1 = 2v_0T^{(0,1)}(-K_1)^{-1}$ and define
    \[ v_{c+1} = v_c \left(T_{++}^{(c,c+1)}+\Psi_c T_{-+}^{(c,c+1)} \right) (-K_{c+1})^{-1}, \]
    for $c=2,\ldots,C$, then the probability $P[\Gamma  = c] = v_c e$ for $c\in \{0,\ldots,C\}$.
\end{itemize}
This follows from the fact that
\[ ¨P[\Gamma = c] = \int_0^\infty \ldots\int_0^\infty [\pi_+(x_1,\ldots,x_c,\vec 0),\pi_-(x_1,\ldots,x_c,\vec 0)] e dx_1 \ldots dx_c,\]
and $\int_0^\infty e^{K_{c_i} x_{c_i}} dx_{c_i} = (-K_{c_i})^{-1}$.
\subsection{The matrices $T_{-+}^{(c)}=0$ for $c < C$}\label{sec:special1}
The second special case does not allow a background transition of the first kind to a state $s'=(x_1,\ldots,x_C,j)$ with $j \not\in S_-$
when the fluid is decreasing, meaning $T_{-+}^{(c)}=0$  for $c \in \{1,\ldots,C-1\}$. In other words, 
any background transition that initiates an increase in the fluid, must be accompanied by a color change. 
For example the background transition occurring at time $3.6$ in Figure \ref{fig:MMFQevo2} is not allowed.
Two simplifications then occur in the computation of the stationary distribution:
\begin{itemize}
    \item The matrices
$K_c$ in \eqref{eq:Kk} are known explicitly for $c < C$ as 
\[K_c = T_{++}^{(c)}.\] 
\item The NAREs in \eqref{eq:NAREs} simplify to
\begin{align}\label{eq:Syvest}
    T_{++}^{(c)}\Psi_c + & \left( T_{+-}^{(c)} + \sum_{\ell > c} T_{++}^{(c,\ell)}\Psi_\ell \right)
   +
   \Psi_c \left(T_{--}^{(c)}+\sum_{\ell > c} T_{-+}^{(c,\ell)}\Psi_\ell \right) = 0,
\end{align}
for $c=1,\ldots,C-1$. These $C-1$ equations are Sylvester matrix equations,
where the $k$-th equation can be solved in $O(|S_+^{(c)}|^3+|S_-|^3)$ time using the
Bartels-Stewart algorithm that makes use of two Schur decompositions \cite{bartels1972solution}.
\end{itemize}

\section{The Colored MMFQ with fluid jumps}\label{sec:jumps}
The background transitions in a (colored) MMFQ never cause an immediate change in the fluid level. This may seem to restrict its applicability, for instance when modeling the workload process of a queue as arrivals cause a jump in the amount of work in the queue. However, these fluid jumps can sometimes be replaced by intervals where the fluid increases linearly to obtain a MMFQ without jumps.
The stationary distribution of the MMFQ with fluid jumps is then obtained from the stationary distribution of the MMFQ without
jumps by censoring out these intervals, see \cite{dzial1,vanHoudt29}. In this section we introduce a general framework for colored MMFQs with fluid jumps
and explain how to reduce them to colored MMFQs introduced in Section \ref{sec:coloredMMFQ}. 
We focus on {\bf upward} fluid jumps, but the same ideas can be used for colored MMFQs with downward fluid jumps.

We assume the upward fluid jump sizes have phase-type (PH) distributions. PH distributions form a general class of distributions that arise as the absorption time of a finite-state continuous-time Markov chain. They are dense in the set of all positive-valued distributions, meaning that any distribution on  $[0,\infty)$  can be approximated arbitrarily closely by a PH distribution \cite{neuts2}. 
The closure properties of PH distributions under convolution, mixtures, and finite minima/maxima further enhance their flexibility \cite{latouche1}. Well known fitting tools and methods for PH distributions include \cite{Horvath2011,Horvath2014,feldman98}.

The state space of a colored MMFQ with fluid jumps is given by
\[ \bar \Omega = \{(0,i)|i\in S_-\} \cup \{\bigcup_{c=1}^C \bar \Omega_c\}\]
with \[ \bar \Omega_c = 
\left\{(x_1,x_2,\ldots,x_C,i) | x_1,\ldots,x_c \geq 0, x_c > 0,
x_{c+1}=\ldots=x_C = 0, i \in S_- \right\}.\]
The background process is always in a state part of the set $S_-$, meaning the fluid decreases at rate $1$ at all times,
unless a background transition occurs.
A colored MMFQ with fluid jumps is characterized by the following matrices
\begin{itemize}
    \item $(T^{(c)}_{--})_{ij}$: rate of moving from state $s \in \bar \Omega_c$ to $s' \in \bar \Omega_c$. These $C$
    matrices have the same meaning as in Section \ref{sec:coloredMMFQ}. The same holds for the matrix $T_{--}^{(0)}$.
    \item $(Q_\ell^{(c,c')})_{ij}$ with $c' > c$: rate of having a {\bf fluid jump} from state $s =(x_1,\ldots,x_c,0,\ldots,0,i) \in \bar \Omega_c$ to 
    \[s'=(x_1,\ldots,x_c,0,\ldots,0,x_{c'},0,\ldots,0,j) \in \bar \Omega_{c'},\] 
    where the value of $x_{c'}$ 
    is drawn from a phase-type (PH) distribution with parameters $(\alpha_{(\ell,c')},U_{(\ell,c')})$, for $\ell=1,\ldots,L_{c'}$. 
    The index $\ell$ represents the {\it type} of the fluid jump and the number of types may depend on $c'$.
    The matrices $Q_\ell^{(0,c')}$ are defined similarly
    for fluid jumps from the boundary states $\{(0,i)|i\in S_-\}$.
    \item $(Q_\ell^{(c)})_{ij}$: rate of having a {\bf fluid jump} from state $s =(x_1,\ldots,x_c,0,\ldots,0,i) \in \bar \Omega_c$ to 
    \[s'=(x_1,\ldots,x_c+y,0,\ldots,0,j) \in \bar \Omega_{c},\] 
    where the value of $y$ 
    is drawn from a PH distribution with parameters $(\alpha_{(\ell,c)},U_{(\ell,c)})$, for $\ell=1,\ldots,L_{c}$. 
\end{itemize}
A colored MMFQ with fluid jumps is fully characterized by the $C+1$ matrices $T_{--}^{(c)}$, the
matrices $Q_\ell^{(c)}$ with $c\in \{1,\ldots,C\}$ and $\ell\in \{1,\ldots,L_{c}\}$, the matrices $Q_\ell^{(c,c')}$ with $c\in \{0,\ldots,C\}$, $c' > c$ and $\ell\in \{1,\ldots,L_{c'}\}$ 
and the PH distributions $(\alpha_{(\ell,c)},U_{(\ell,c)})$
for $c=1,\ldots,C$ and $\ell=1,\ldots,L_c$. Notice that when a fluid jump occurs, the fluid that is added has 
one color. This assumption can be relaxed by allowing that fluid of different colors
is added in a single fluid jump. We restrict ourselves to fluid jumps of one color as this suffices for the applications considered in Section \ref{sec:apps}. In fact in our
applications we always have that $c'=c$ or $c+1$ when a fluid jump occurs, so colors are never skipped.

\begin{figure}
\centering
\begin{minipage}{0.35\textwidth}
\begin{tikzpicture}[x=1.2cm,y=1cm,>=Latex,
  axis/.style={thin},
  inc/.style={very thick},             % slope +1
  dec/.style={very thick,dashed},      % slope -1
  lab/.style={font=\small}]
% --- editable parameters ---
\def\xstart{2} % initial level
\def\dA{0.7}  % state 1  ( +1 )
\def\jA{1.5}
\def\dB{0.5}  % state 3  ( -1 )
\def\dC{0.6}  % state 2  ( +1 )
\def\jB{0.5}
\def\dD{1.2}  % state 3  ( -1 )

% --- starting point ---
\coordinate (P0) at (0,\xstart);

% --- trajectory segments ---
\draw[dec] (P0) -- node[above,sloped,lab]{-1} ++(\dA,-\dA) coordinate (P2);
\draw[inc] (P2) -- node[above,sloped,lab]{} ++(0,\jA) coordinate (P3);
\draw[dec] (P3) -- node[above,sloped,lab]{-1} ++(\dB,-\dB) coordinate (P4);
\draw[dec] (P4) -- node[above,sloped,lab]{-1} ++(\dC,-\dC) coordinate (P6);
\draw[inc] (P6) -- node[above,sloped,lab]{} ++(0,\jB) coordinate (P7);
\draw[dec] (P7) -- node[above,sloped,lab]{-1} ++(\dD,-\dD) coordinate (P8);

\path[fill=black!10]
  (P0 |- 0,0) -- (P0) --  (P2) -- ($(P2 -| \dA+\jA+\jB,0)$) -- (P8) -- (P8 |- 0,0) -- cycle;

\path[fill=red!10]
 (P2) -- (P3) -- (P6) -- ($(P6 -| \dA+\dB+\dC+\jB,0)$) -- ($(P2 -| \dA+\jA+\jB,0)$) -- cycle;

\path[fill=green!10]
 (P6) -- (P7) -- ($(P6 -| \dA+\dB+\dC+\jB,0)$)  -- cycle;

\draw[dotted] (P6) -- ($(P6 -| \dA+\dB+\dC+\jB,0)$);
\draw[dotted] ($(P6 -| \dA+\dB+\dC+\jB,0)$) -- ($(0,0 -| \dA+\dB+\dC+\jB,0)$);
\draw[dotted] (P2) -- ($(P2 -| \dA+\jA+\jB,0)$);
\draw[dotted] ($(P2 -| \dA+\jA+\jB,0)$) -- ($(0,0 -| \dA+\jA+\jB,0)$);

% --- axes (extend to last point) ---
\draw[axis,->] (0,0) -- ($(P8 |- 0,0) + (0.6,0)$) node[below right] {$t$};
\draw[axis,->] (0,0) -- (0,5) node[left] {};

% initial marker & label
%\fill (P0) circle[radius=1.2pt];
%\node[lab,anchor=west] at ($(P0)-(0.06,0.18)$) {$x(0)=\xstart$};

% --- light guides at switch times (optional, remove if not needed) ---
\foreach \Q in {P2,P4,P6}
  \draw[axis,dashed] (\Q |- 0,0) -- (\Q |- 0,4.9);

% --- state labels at y=4 between guides ---
\node[lab] at ($(P0 |- 0,4)!0.5!(P2 |- 0,4)$) {1};
\node[lab,blue] at ($(P2 |- 0,4)!0.5!(P4 |- 0,4)$) {3};
\node[lab] at ($(P4 |- 0,4)!0.5!(P6 |- 0,4)$) {2};
\node[lab,blue] at ($(P6 |- 0,4)!0.5!(P8 |- 0,4)$) {2};

\tikzset{wobbly/.style={decorate, decoration={snake, amplitude=1mm, segment length=5mm}, thick}}

\draw[wobbly] ($(P8 |- 0,2.8) - (0.5,0)$) -- ($(P8 |- 0,2.8) + (0.7,0)$) node[midway,above]{};
\draw[wobbly] ($(P8 |- 0,3) - (0.5,0)$) -- ($(P8 |- 0,3) + (0.7,0)$) node[midway,above]{};
\path[fill=black]
 ($(P8 |- 0,2.8) + (0.7,-0.2)$) -- ($(P8 |- 0,3) + (0.7,0.2)$) -- ($(P8 |- 0,2.9) + (1,0)$)  -- cycle;

%\node[lab, anchor=west,text width=2cm] at ($(P6 |- 0,4) + (0.3,0)$) {background state};
%\node[lab, anchor=west] at ($(P6) + (0.3,0)$) {fluid level};

\node[lab, anchor=north] at ($(P0 |- 0,0)$) {0};
\node[lab, anchor=north] at ($(P2 |- 0,0)$) {\dA};
\node[lab, anchor=north] at ($(P4 |- 0,0)$) {$\pgfmathparse{\dA+\dB}\pgfmathprintnumber{\pgfmathresult}$};
\node[lab, anchor=north] at ($(P6 |- 0,0)$) {$\pgfmathparse{\dA+\dB+\dC}\pgfmathprintnumber{\pgfmathresult}$};
\node[lab, red, anchor=north] at ($(\dA+\jA+\jB,0 |- 0,0)$) {$\pgfmathparse{\dA+\jA+\jB}\pgfmathprintnumber{\pgfmathresult}$};
\node[lab, red, anchor=north] at ($(\dA+\dB+\dC+\jB,0 |- 0,0)$) {$\pgfmathparse{\dA+\dB+\dC+\jB}\pgfmathprintnumber{\pgfmathresult}$};
\end{tikzpicture}
\end{minipage}
\hfill
\begin{minipage}{0.6\textwidth}
\begin{tikzpicture}[x=1.2cm,y=1cm,>=Latex,
  axis/.style={thin},
  inc/.style={very thick},             % slope +1
  dec/.style={very thick,dashed},      % slope -1
  lab/.style={font=\small}]
% --- editable parameters ---
\def\xstart{2} % initial level
\def\dA{0.7}  % state 1  ( +1 )
\def\jAA{0.9}
\def\jAB{0.6}
\def\jA{1.5}
\def\dB{0.5}  % state 3  ( -1 )
\def\dC{0.6}  % state 2  ( +1 )
\def\jB{0.5}
\def\dD{1.2}  % state 3  ( -1 )

% --- starting point ---
\coordinate (P0) at (0,\xstart);

% --- trajectory segments ---
\draw[dec] (P0) -- node[above,sloped,lab]{-1} ++(\dA,-\dA) coordinate (P2);
\draw[inc] (P2) -- node[above,sloped,lab]{+1} ++(\jAA,\jAA) coordinate (P3);
\draw[inc] (P3) -- node[above,sloped,lab]{+1} ++(\jAB,\jAB) coordinate (P3a);
\draw[dec] (P3a) -- node[above,sloped,lab]{-1} ++(\dB,-\dB) coordinate (P4);
\draw[dec] (P4) -- node[above,sloped,lab]{-1} ++(\dC,-\dC) coordinate (P6);
\draw[inc] (P6) -- node[above,sloped,lab]{+1} ++(\jB,\jB) coordinate (P7);
\draw[dec] (P7) -- node[above,sloped,lab]{-1} ++(\dD,-\dD) coordinate (P8);

\path[fill=black!10]
  (P0 |- 0,0) -- (P0) --  (P2) -- ($(P2 -| \dA+\jA+\jA+\jB+\jB,0)$) -- (P8) -- (P8 |- 0,0) -- cycle;

\path[fill=red!10]
 (P2) -- (P3a) -- (P6) -- ($(P6 -| \dA+\jA+\dB+\dC+\jB+\jB,0)$) -- ($(P2 -| \dA+\jA+\jA+\jB+\jB,0)$) -- cycle;

\path[fill=green!10]
 (P6) -- (P7) -- ($(P6 -| \dA+\jA+\dB+\dC+\jB+\jB,0)$)  -- cycle;

\draw[dotted] (P6) -- ($(P6 -| \dA+\jA+\dB+\dC+\jB+\jB,0)$);
\draw[dotted] ($(P6 -| \dA+\jA+\dB+\dC+\jB+\jB,0)$) -- ($(0,0 -| \dA+\jA+\dB+\dC+\jB+\jB,0)$);
\draw[dotted] (P2) -- ($(P2 -| \dA+\jA+\jA+\jB+\jB,0)$);
\draw[dotted] ($(P2 -| \dA+\jA+\jA+\jB+\jB,0)$) -- ($(0,0 -| \dA+\jA+\jA+\jB+\jB,0)$);

% --- axes (extend to last point) ---
\draw[axis,->] (0,0) -- ($(P8 |- 0,0) + (0.6,0)$) node[below right] {$t$};
\draw[axis,->] (0,0) -- (0,5) node[left] {};

% initial marker & label
%\fill (P0) circle[radius=1.2pt];
%\node[lab,anchor=west] at ($(P0)-(0.06,0.18)$) {$x(0)=\xstart$};

% --- light guides at switch times (optional, remove if not needed) ---
\foreach \Q in {P2,P3,P3a,P4,P6,P7}
  \draw[axis,dashed] (\Q |- 0,0) -- (\Q |- 0,4.9);

% --- state labels at y=4 between guides ---
\node[lab] at ($(P0 |- 0,4)!0.5!(P2 |- 0,4)$) {1};
\node[lab] at ($(P2 |- 0,4)!0.5!(P3 |- 0,4)$) {\tiny (3,(1,1))};
\node[lab] at ($(P3 |- 0,4)!0.5!(P3a |- 0,4)$) {\tiny (3,(1,2))};
\node[lab,blue] at ($(P3a |- 0,4)!0.5!(P4 |- 0,4)$) {3};
\node[lab] at ($(P4 |- 0,4)!0.5!(P6 |- 0,4)$) {2};
\node[lab] at ($(P6 |- 0,4)!0.5!(P7 |- 0,4)$) {\tiny (2,(2,1))};
\node[lab,blue] at ($(P7 |- 0,4)!0.5!(P8 |- 0,4)$) {2};

\draw[lab,->,blue] ($(P3a |- 0,4)!0.5!(P4 |- 0,4)+(0,0.3)$) -- ++(0,0.25) -- ($(P2 |- 0,4)!0.5!(P3 |- 0,4)+(-0.15,0.55)$) -- ++(0,-0.3);
\draw[lab,->,blue] ($(P3a |- 0,4)!0.5!(P4 |- 0,4)+(0,0.3)$) -- ++(0,0.25) -- ($(P3a |- 0,4)!0.5!(P3 |- 0,4)+(-0.15,0.55)$) -- ++(0,-0.3);
\draw[lab,->,blue] ($(P7 |- 0,4)!0.5!(P8 |- 0,4)+(0,0.3)$) -- ++(0,0.25) -- ($(P6 |- 0,4)!0.5!(P7 |- 0,4)+(-0.15,0.55)$) -- ++(0,-0.3);

%\node[lab, anchor=west,text width=2cm] at ($(P6 |- 0,4) + (0.3,0)$) {background state};
%\node[lab, anchor=west] at ($(P6) + (0.3,0)$) {fluid level};

\node[lab, anchor=north] at ($(P0 |- 0,0)$) {0};
\node[lab, anchor=north] at ($(P2 |- 0,0)$) {\dA};
\node[lab, anchor=north] at ($(P3 |- 0,0)$) {$\pgfmathparse{\dA+\jAA}\pgfmathprintnumber{\pgfmathresult}$};
\node[lab, anchor=north] at ($(P3a |- 0,0)$) {$\pgfmathparse{\dA+\jA}\pgfmathprintnumber{\pgfmathresult}$};
\node[lab, anchor=north] at ($(P4 |- 0,0)$) {$\pgfmathparse{\dA+\jA+\dB}\pgfmathprintnumber{\pgfmathresult}$};
\node[lab, anchor=north] at ($(P6 |- 0,0)$) {$\pgfmathparse{\dA+\dB+\dC+\jA}\pgfmathprintnumber{\pgfmathresult}$};
\node[lab, anchor=north] at ($(P7 |- 0,0)$) {$\pgfmathparse{\dA+\dB+\dC+\jA+\jB}\pgfmathprintnumber{\pgfmathresult}$};

\node[lab, red, anchor=north] at ($(\dA+\jA+\jB+\jA+\jB,0 |- 0,0)$) {$\pgfmathparse{\dA+\jA+\jB+\jA+\jB}\pgfmathprintnumber{\pgfmathresult}$};
\node[lab, red, anchor=north] at ($(\dA+\dB+\dC+\jB+\jA+\jB,0 |- 0,0)$) {$\pgfmathparse{\dA+\dB+\dC+\jB+\jA+\jB}\pgfmathprintnumber{\pgfmathresult}$};
\end{tikzpicture}
\end{minipage}

    \caption{Illustration of the reduction of a sample path of a $3$-colored MMFQ {\it with fluid jumps} to a sample path
    of a $3$-colored MMFQ.  $S_-=\{1,2,3\}$ and $S^{(c)}_+=\{1,2,3\} \times \{(1,1),(1,2),(2,1)\}$ if type-2 fluid jumps
    have an order $2$ phase-type representation and type-2 jumps are exponential.
    The fluid jump occurring at time $0.7$ is a type-$1$ jump of size $1.5$, the fluid jump at time $1.8$ is type-$2$ and has size
    $0.5$.}
    \label{fig:MMFQreduce}
\end{figure}
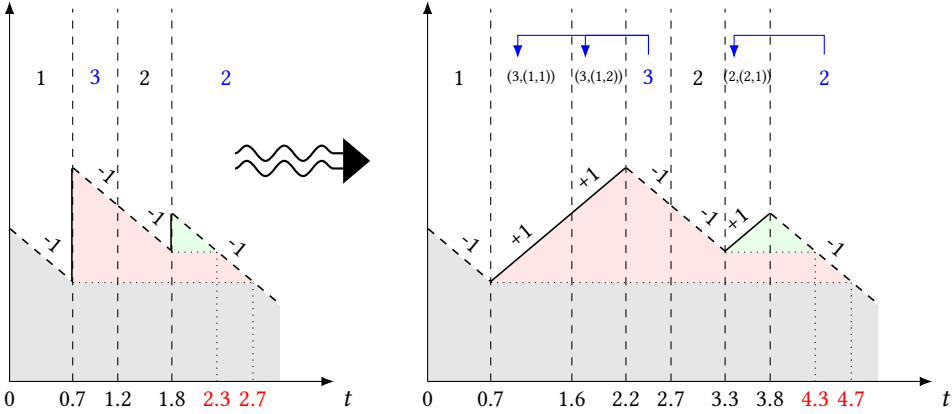

To analyze a colored MMFQ with fluid jumps, we replace the fluid jumps of height $h$ by an interval of length
$h$ where the fluid is added at rate $1$ (see Figure \ref{fig:MMFQreduce}). After computing the stationary distribution of the resulting colored MMFQ using
Theorem \ref{th:main}, we censor out these intervals to obtain the stationary distribution of the colored MMFQ with fluid jumps.
This is essentially possible because of the phase-type nature of the fluid jumps and would not work with 
general fluid jump size distributions. 

The set $S_-$ is the same for the colored MMFQ with and without fluid jumps. We set $S_+^{(c)} = S_- \times 
\{(\ell,m)| \ell=1,\ldots,L_c, m=1,\ldots, M_{(\ell,c)}\}$, where
$M_{(\ell,c)}$ is the order of the PH distribution $(\alpha_{(\ell,c)},U_{(\ell,c)})$. 
A state $j=(a,(\ell,m))\in S_+^{(c)}$ has two
components: 
\begin{itemize}
    \item A component $a$ that represents the state of the background process of the colored MMFQ with fluid jumps
    immediately after the fluid jump. In the colored MMFQ without fluid jumps, $a$ remains frozen
    as long as fluid is added.
    \item A component $(\ell,m)$, where $m$ represents the phase of the fluid that is being added. The phase is initialized by $\alpha_{(\ell,c)}$ 
    when fluid of color $c$ is added by a type-$\ell$ fluid jump and then evolves according to the matrix $U_{(\ell,c)}$ until the required amount
    of fluid is added. 
\end{itemize}
Let us now make this more precise by introducing the matrices that characterize the colored MMFQ. 
\begin{itemize}
    \item At rate $(Q_\ell^{(c,c')})_{ij}$ 
    an upward movement in the fluid level from state $(x_1,\ldots,x_c,0\ldots,0,i)$ is initiated 
    associated to a type-$\ell$ fluid jump and 
    the colored MMFQ enters state 
    \[(x_1,\ldots,x_n,0,\ldots,0,0^+,0\ldots,0,(j,(\ell,m))),\] 
    with probability $(\alpha_{(\ell,c')})_m$, where the $0^+$ occurs in position $c'$
    and $j \in S_+^{(c')}$. 
    We therefore have in matrix form
    \[ T_{-+}^{(c,c')}= [Q_1^{(c,c')} \otimes \alpha_{(1,c')}, \ Q_2^{(c,c')} \otimes \alpha_{(2,c')}, \ \ldots,\ Q_{L_{c'}}^{(c,c')} \otimes \alpha_{(L_{c'},c')}] \]
    for $0 \leq c < C$.
    Note that $(\ell,m)$ is the initial state
    of the PH distribution characterized by $(\alpha_{(\ell,c)},U_{(\ell,c)})$ 
    that is used to add the work of the type-$\ell$ fluid jump to the fluid queue.
    The background state $j$ is stored in the first entry of $(j,(\ell,m))$ (see blue arrows in Figure \ref{fig:MMFQreduce}). This state $j$ remains frozen as long as the fluid increases. 
    \item Completely analogous, we have 
    \[ T_{-+}^{(c)}= [Q_1^{(c)} \otimes \alpha_{(1,c)}, \ Q_2^{(c)} \otimes \alpha_{(2,c)}, \ \ldots,\ Q_{L_{c}}^{(c)} \otimes \alpha_{(L_{c},c)}] \]
    for $1 \leq c \leq C$. 
    \item The color $c$ fluid associated to a type-$\ell$ fluid jump has a PH distribution characterized by $(\alpha_{(\ell,c)},U_{(\ell,c)})$, so work is added according to
    the matrix $U_{(\ell,c)}$, meaning
    \[ T_{++}^{(c)} = I \otimes \begin{bmatrix}
        U_{(1,c)} & & \\
        & \ddots & \\
        & & U_{(L_c,c)}
    \end{bmatrix},\]
    for $1 \leq c \leq C$, where $I$ is the identity matrix of size $|S_-|$. 
    The vector $(-U_{(\ell,c)})e$ holds the rates at which the increasing interval
    of the fluid queue ends:
    \[ T_{+-}^{(c)} = I \otimes \begin{bmatrix}
        (-U_{(1,c)})e \\ \vdots \\ (-U_{(L_c,c)})e
    \end{bmatrix},\]
    for $1 \leq c \leq C$.
    \item The matrices $T_{++}^{(c,c')}=0$ for $c' > c$, as any upward jump is followed by an interval where
    the fluid decreases. This is due to our assumption that upward fluid jumps add one color.
\end{itemize}
The stationary distribution of the colored MMFQ {\it with fluid jumps} 
(we add a superscript $^{jumps}$ to distinguish with the colored MMFQ without fluid jumps) 
can now be readily obtained from Theorem \ref{th:main} as
\begin{itemize}
    \item  $\pi_-^{jumps}(\vec x)$ is given by the right-hand side of \eqref{eq:pix} if we replace
$[I,\Psi_{k_n}]$ by $\Psi_{k_n}$, because the stationary distribution of a Markov process censored on a subset of the state space
is proportional to the stationary distribution of the uncensored process.
\item Equation \eqref{eq:pmin} remains valid for the same reason, that is, $p_-^{jumps}$ also
solves \eqref{eq:pmin} and
\item the normalization of $p_-^{jumps}$ is the same as in \eqref{eq:norm} if we remove the $2$, because the amount of time
that the fluid goes up equals the amount of time that it goes down.
\end{itemize}
The expression for $\Xi^{jumps}$ and $\Gamma^{jumps}$ for the colored MMFQ with fluid jumps is the same as in 
\eqref{eq:Xi} and \eqref{eq:Gamma}, respectively, if we remove the $2$ appearing on the right hand side.

\section{Applications}\label{sec:apps}
In this section we present a number of applications. These concern queueing systems that
cannot be analyzed in an efficient manner using existing Markov chain modeling approaches. 
For each of these applications we introduce a colored MMFQ with fluid jumps, which can be analyzed by
constructing a colored MMFQ without fluid jumps as explained in Section \ref{sec:jumps}.

\subsection{The MMAP[L]/PH[L]/1/N/LCFS queue}\label{sec:example1}
We start with a  queueing model that is further generalized in the next subsection and has the following properties:
\begin{itemize}
    \item The queue has a single server and serves jobs in Last-Come-First-Served (LCFS) order with job resumption. This means that the work performed
    on an interrupted job is not lost.  
    \item The queue can hold up to $N$ jobs: one in the server and $N-1$ in the waiting room.
    \item Jobs arrive according to an MMAP[L] arrival process \cite{asmussen1993marked,HE4}. Such a process is a multi-type version of the common
    MAP arrival process. It is characterized by a set of $L+1$ square matrices of size $M_a$: $D_0,D_1,\ldots,D_L$, where 
    $(D_\ell)_{ij}$, for $\ell =1,\ldots,L$, is the rate at which type-$\ell$ arrivals occur in state $i$ of the arrival process, while the state changes to state $j$. $(D_0)_{ij}$, with $i \not= j$, is the rate at which a state transition occurs from state $i$ to state $j$     without an arrival. Finally $-(D_0)_{ii} = \sum_{j \not= i}(D_0)_{ij}+\sum_{\ell=1}^L \sum_{j} (D_\ell)_{ij}$.
    Note that consecutive inter-arrival times 
    and consecutive job types can be correlated, and there can also be correlation between inter-arrival times and job types.  
    Fitting algorithms for MMAP[L] arrival processes (and subclasses thereof) have been proposed in \cite{Buchholz2010,Horvath2013,Casale2016}.
    \item The time needed to process a job has a PH distribution that depends on its type, this is denoted as PH[L]
    service. The service time $X_\ell$ of a type-$\ell$ job is characterized by $(\beta_\ell,V_\ell)$ such that $P[X_\ell > t] = \beta_\ell e^{V_\ell t} e$, where $V_\ell$ is a size $M_{s_\ell}$ matrix. For instance, if type-$\ell$ jobs have exponential durations with parameter
    $\mu_\ell$, then $\beta_\ell = 1$ and $V_\ell = -\mu_\ell$.
\end{itemize}
MMAP[L]/PH[L]/1/FCFS queues can be analyzed using an age or workload process \cite{vanHoudt29,HE9}, while
the analysis of MMAP[L]/PH[L]/1/LCFS queues relies on tree-structured Quasi-Birth-Death Markov chains \cite{HE6}. However these approaches no longer work when the queue
has a finite capacity. In this section we show how we can deal with MMAP[L]/PH[L]/1/N/LCFS queues by constructing
a colored MMFQ such that {\bf the computation time of the queue length distribution and loss probability 
is linear in $N$} and cubic in $M_a \sum_\ell M_{s_\ell}$. Note that if we simply define a finite state Markov chain by keeping
track of the content of the MMAP[L]/PH[L]/1/N/LCFS queue, the number of states would be exponential in $N$ (as we need to keep
track of the types of the jobs in the queue).

We now introduce the colored MMFQ with fluid jumps. The idea is the following: 
the fluid of the queue reflects the workload in the queue; hence, the workload jumps up whenever
an arrival occurs. We use color $c$ to represent the work associated with
the $c$-th oldest job in the queue, meaning we use $C=N$ colors. To deal with the finiteness of the queue, new job arrival
events are blocked when the color on top of the fluid is color $N$.

When there are $n \in\{1,\ldots,N\}$ jobs in the queue, the colored MMFQ with fluid jumps is in some state $(x_1,\ldots,x_N,j)$ where $x_1,\ldots,x_n >0$ and $x_{n+1},\ldots,x_N = 0$. The value of $x_i$ reflects
the remaining amount of work for the $i$-th oldest job in the queue (for $i \leq n)$.  The background process
of the colored MMFQ with fluid jumps has $|S_-|=M_a$ states that correspond to the states of the arrival process. 
The colored MMFQ with jumps evolves as follows
when there are $n < N$ jobs in the queue:
\begin{itemize}
    \item Unless a background transition occurs the workload decreases at rate one possibly changing the color at the top of the fluid (this happens when a job
    completes service).
    \item At rate $(D_0)_{ij}$, with $i\not=j$, the background state changes from $i$ to $j$ while the fluid continues to decrease.
    \item At rate $(D_\ell)_{ij}$ a type-$\ell$ arrival occurs, this causes a fluid jump that adds an amount
    of fluid with color $n+1$ that corresponds to the size of the incoming type-$\ell$ job. 
\end{itemize}
Let us now specify the matrices $T_{--}^{(c)}$,  $Q_\ell^{(c)}$ and $Q_\ell^{(c,c')}$, as well as the PH distributions $(\alpha_{(\ell,c)},U_{(\ell,c)})$:
\begin{itemize}
    \item For  $0 \leq c < N$ we have $T_{--}^{(c)} = D_0$, that is, the fluid decreases at rate $1$ and the state of the arrival process
    may change without causing a job arrival (e.g., a time $1.2$ in Figure \ref{fig:MMFQreduce}). 
    When the queue is full, new jobs are discarded/blocked, but the
    phase of the arrival process may still change, therefore, $T_{--}^{(N)} = D_0 + \sum_{\ell=1}^L D_\ell$.
    \item At rate $(D_\ell)_{ij}$ a type-$\ell$ arrival occurs in state $(0,i)$ or $(x_1,\ldots,x_c,0\ldots,0,i)$, with $0 < c < N$. We therefore have in matrix form
    \[ Q_\ell^{(c,c+1)}= D_\ell, \]
    for $0 \leq c < N$ and $\ell=1,\ldots,L_c=L$.
    For the PH distributions we have $(\alpha_{(\ell,c)},U_{(\ell,c)})=(\beta_\ell,V_\ell)$, meaning the distribution only depends on $\ell$, and
    not on $c$ (as the amount of work of a job depends on its type $\ell$, but not on the number of jobs waiting in the queue).
\end{itemize}
The $Q_\ell^{(c)}$ matrices and the remaining $Q_\ell^{(c,c')}$ matrices  are zero for all $\ell$.
It is worth noting that for the associated colored MMFQ without fluid jumps,  both
the special case of subsection \ref{sec:special0} and subsection \ref{sec:special1} apply. 
Moreover, the matrix $T_{-+}^{(N)}=0$ which means that
$\Psi_N$ can also be determined by solving a Sylvester matrix equation. As a result all the $\Psi_c$ and $K_c$ matrices 
as well as the queue length distribution can be computed in $O(M_a^3 (\sum_\ell M_{s_\ell})^3N)$ time.
For this application, the distribution of $\Gamma^{jumps}$ of the colored MMFQ with fluid jumps corresponds 
to the queue length distribution.

\subsection{The MMAP[L]/PH[L]/1/N[L]/LCFS queue}
In this subsection we consider a generalization of the model in the previous subsection. 
In the  MMAP[L]/PH[L]/1/N/LCFS queue all job types
are accepted as long as the queue contains less than $N$ jobs and all jobs are rejected when there are $N$ jobs in the queue.
We now relax this by associating a threshold $N_\ell$ to each job type $\ell$, for $\ell=1,\ldots,L$. The idea is that type-$\ell$ jobs
are accepted as long as the number of jobs in the queue {\bf  of any type} is below $N_\ell$, while type-$\ell$ jobs are  rejected as soon as there
are $N_\ell$ jobs {\bf of any type} in the queue. We allow that some of the $N_\ell$ values equal $+\infty$, but at least one value is finite.
Otherwise the queue is an MMAP[L]/PH[L]/1/LCFS queue which can be analyzed with existing methods \cite{HE3}.

We restrict ourselves to discussing the main changes required to the colored MMFQ with fluid jumps
presented in the previous subsection. 
Without loss of generality assume the job types are ordered such that $1 \leq N_L \leq \ldots \leq N_2 \leq N_1 \leq +\infty$.
The following changes are required compared to subsection \ref{sec:example1}, which corresponded to setting $N_\ell = N < +\infty$ 
for $\ell=1,\ldots,L$:
\begin{itemize}
    \item The number of colors used equals $C=N_1$ if $N_1$ is finite and $C=1+\max  \{ N_\ell | N_\ell < +\infty\}$ otherwise.
    For instance, if $(N_1,N_2,N_3)=(+\infty,20,10)$, then
    $C=21$.
    \item If there are $c \leq C$ jobs in the queue, jobs of types $\{\ell | N_\ell \leq c\}$ are rejected; therefore
    \[ T_{--}^{(c)} = D_0 + \sum_{\ell:  N_\ell \leq c} D_\ell,\]
    for $0 \leq c\leq C$. Note that   
    $T_{--}^{(C)} = D_0+\sum_{\ell:N_\ell < +\infty} D_\ell$.  
    If $(N_1,N_2,N_3)=(+\infty,20,10)$, then $T_{--}^{(20)}= T_{--}^{(21)}=D_0+D_1+D_2$ and $T_{--}^{(19)}=D_0+D_1$.
    \item As only type-$\ell$ jobs with $N_\ell > c$ can be added
    to the queue as the $(c+1)$-th job, only type-$1$ to type-$L_{c+1}$ jobs can arrive with $L_{c+1} = \max \{\ell | N_\ell > c\}$
    when the queue length equals $c$.
    If $(N_1,N_2,N_3)=(+\infty,20,10)$, then
    $L_{c+1}=3$, for $c< 10$, $L_{c+1}=2$, for $10 \leq c < 20$ and $L_{c+1}=1$, for $c = 20$.
    As before we have 
    \[Q_\ell^{(c,c+1)}=D_\ell,\] 
    for $\ell = 1,\ldots,L_{c+1}$.
    As in subsection \ref{sec:example1} the PH distributions are such that $(\alpha_{(\ell,c)},U_{(\ell,c)})=(\beta_\ell,V_\ell)$.
    \item When $N_1=+\infty$, we also have matrices $Q_\ell^{(C)}=D_\ell$ for all $\ell$ with $N_\ell=+\infty$. 
%    For the matrices we get
%    \[ T_{-+}^{(n-1,n)} = [D_1 \otimes \alpha_{1}, \ldots, D_{k(n)} \otimes \alpha_{k(n)}],\]
 %   where $k(n)$ is the largest $k$ such that $N_k \geq n$. If $(N_1,N_2,N_3)=(+\infty,20,10)$, then
 %   $k(n)=3$, for $n\leq 10$, $k(n)=2$, for $10< n \leq 20$ and $k(n)=1$, for $n = C = 21$.
 %   The matrices $T_{++}^{(n)}$ and $T_{+-}^{(n)}$ are also reduced in size in the same vein;
 %\[ T_{++}^{(n)} = I \otimes \begin{bmatrix}
 %       V_{1} & & \\
 %       & \ddots & \\
 %       & & V_{k(n)}
 %   \end{bmatrix}, \ \ \ \mbox{ and } \ \ \ \  T_{+-}^{(n)} = I \otimes \begin{bmatrix}
 %       (-V_1)e \\ \vdots \\ (-V_{k(n)})e
 %   \end{bmatrix},\]
 %   for $1 \leq n \leq C$. When $N_K = +\infty$, we also have a matrix
  %  \[ T_{-+}^{(C)} = [D_{k(C)} \otimes \alpha_{k(C)}, \ldots, D_{K} \otimes \alpha_K],\]
  %  which deals with arrivals of the job types with an infinite threshold.
\end{itemize}
This model also fits within the special cases of subsection \ref{sec:special0} and \ref{sec:special1}. 
When $N_1 < +\infty$ the matrix $T_{-+}^{(C)}$ is zero and equation \eqref{eq:NAREC} for $\Psi_C$ also reduces to 
a Sylvester matrix equation; for  $N_1 = +\infty$, it does not.

\subsection{A finite FCFS queue with a "multi-level job" cascade}\label{sec:apps3}
The third example is a queue with FCFS service, it has the following properties:
\begin{itemize}
    \item The queue has a single server and can hold up to $N$ level-$1$ jobs.  
    \item Level-$1$ jobs arrive exogenously according to a MAP process characterized by the matrices $D_0$ and $D_1$
    (equivalently, an MMAP[L] with $L=1$, see Section \ref{sec:example1} for details).
    \item There are $C$ job levels that are all served by the same single server. While a level-$c$ job is in service (for $c=1,\dots,C-1$), level-$(c+1)$ jobs are generated by a Poisson process of rate $\gamma_c$. 
        A level-$(c+1)$ job $y$ is called a \emph{child} of a level-$c$ job $x$ if $y$ is generated during the service of $x$. 
        A level-1 job departs the system only after its own service \emph{and} the service of all jobs in its descendant tree have completed.
    \item Level-1 jobs are served in FCFS order, and a level-1 job may begin service only after the previous level-1 job has left the system (i.e., all of its descendants have completed). 
        The order in which the descendants of a job currently in service are processed is at the scheduler’s discretion (e.g., FCFS, LIFO, priority by level), provided the server is work-conserving.
    \item The service time of a level-$c$ job follows a PH distribution with parameters $(\beta_c,V_c)$.
\end{itemize}
Our objective is to determine the queue length distribution of the level-$1$ jobs. We can therefore select the order
in which the descendants of a level-$1$ job are served in a convenient manner, that is, we assume these jobs are
served using preemptive priority where a level-$(c+1)$ job has higher priority than a level-$c$ job. In other words
the tree of descendants is traversed in depth-first order with preemption.

We can model this queue as a classic MAP/PH/1/N queue. However, the number of
PH phases required would be exponential in $C$ as the level $1$ to $c$ service phases at preemption time need to be
stored when a level-$(c+1)$ job is in service. If we want to avoid preemption, we are faced with the problem
that there is no upper bound on the number of level-$(c+1)$ jobs that can be generated during the service of a level-$c$
job, meaning the PH distribution would have an infinite number of states. 

We now illustrate how this queueing system can be analyzed in an efficient manner using a colored MMFQ with fluid jumps. 
\begin{itemize}
    \item The set of states $S_- = \{1,\ldots,M_a\} \times \{1,\ldots,N\}$ and keeps track of the MAP state and
    the number of level-$1$ jobs in the queue. 
    \item We use $C$ colors. If a level-$c$ job is in service, the top color is $c$ and 
    for $k \in \{1,\ldots,c\}$ the amount of fluid of
    color $k$ is the remaining service time of the level-$k$ job (that was interrupted if $k < c$).
    Note that as long as the fluid level remains larger than zero, all jobs in service are descendants of
    the same level-$1$ job.
    \item At rate $(D_1)_{ij}$ an arrival occurs, this changes the background state from $(i,n)$ to $(j,\min(N,n+1))$, while
    at rate $(D_0)_{ij}$ the background state changes from $(i,n)$ to $(j,n)$. Hence,
    \[ T^{(c)}_{--} = D_0 \otimes I + D_1 \otimes \begin{bmatrix}
        0 & 1 & &  \\
       & \ddots& \ddots & \\
        & & 0 & 1\\
        & &  & 1
    \end{bmatrix} \]
    for $c=1\ldots,C$.
    \item  When a level-$c$ job is in service, meaning the color at the top of the fluid is $c$, 
    a level-$(c+1)$ job is generated at rate $\gamma_c$. This causes an immediate upward jump in the fluid level, the size of which has a PH distribution with parameters $(\beta_{c+1},V_{c+1})$. Hence,
    \[ Q_1^{(c,c+1)} = \gamma_c I,\]
    with $L_c=1$, for all $c$, and $(\alpha_{(1,c)},U_{(1,c)}) = (\beta_c,V_c)$.
\end{itemize}
There is, however, a complication that we disregarded so far. When the fluid level hits zero and the background state is $(i,n)$ with $n > 1$, the colored MMFQ with fluid jumps should immediately add the work of the next level-$1$ job in the queue and the background state should become $(j,n-1)$. It may appear at first glance that frameworks introduced in Sections \ref{sec:coloredMMFQ} 
and \ref{sec:jumps} can therefore not be used  as they do not support such flexibility at the boundary when the fluid becomes zero.
There is however a simple solution for this: whenever the fluid hits zero with a background state of the form
$(i,n)$ with $n > 1$, we simply remain within this state for an exponential amount of time with mean one and
then jump up. Note that during this time the state $i$ of the MAP is frozen. In other words the size $M_aN$ matrix $T_{--}^{(0)}$
is given by
\[ T_{--}^{(0)} = \begin{bmatrix}
    D_0 & 0 &\ldots & 0\\
    0 & -I &\ldots & 0\\
    \vdots & \ddots &\ddots & 0\\
    0 &\ldots & 0 & -I\\
\end{bmatrix} \quad \mbox{ and } \quad  Q_1^{(0,1)} =  \begin{bmatrix}
    D_1 & 0 &\ldots & 0\\
    I & 0 &\ldots & 0\\
    \vdots & \ddots &\ddots & 0\\
    0 &\ldots & I & 0\\
\end{bmatrix}.\]
After obtaining the stationary distribution of the colored MMFQ without jumps, we censor out all the
periods where the fluid increases as well as the periods where the level is equal to zero with
a background state of the form $(i,n)$ with $n> 1$.

\section{Numerical and runtime results}\label{sec:num}

In this section we first present some fairly arbitrary numerical examples to demonstrate that the colored MMFQ with
fluid jumps can be used to efficiently compute the queue length distribution and loss probabilities in an
MMAP[L]/PH[L]/1/N[L]/LCFS queue. The computation time required for these examples is a fraction of a second,
even with $N_1=1000$, while the finite state Markov chain that keeps track of the content of the queue would
require $\Theta(2^{N_1})$ states. Next we present runtime results for the application in Section \ref{sec:apps3},
where we compare the runtime results of the colored MMFQ framework with the classic Quasi-Birth-Death Markov chain used to
solve the MAP/PH/1/N queue.

\subsection{Numerical results for the MMAP[L]/PH[L]/1/N[L]/LCFS queue}
The MMAP[L] arrival process with arrival rate $\lambda$ is characterized by
\begin{align}
D_0= \begin{bmatrix}
 -\lambda-1/q_1 &1/q_1 \\ 1/q_2 &-\lambda-1/q_2    
\end{bmatrix}, \ \ 
D_1= \begin{bmatrix}
 \lambda p_1 & 0 \\ 0 &\lambda p_2    
\end{bmatrix}, \ \ 
D_2= \begin{bmatrix}
 \lambda (1-p_1) & 0 \\ 0 &\lambda (1-p_2)    
\end{bmatrix},
\end{align}
with $q_1 = 100, q_2 = 500$, $p_1 = 0.1$ and $p_2 = 0.3$. 
Type-1 jobs are exponential with mean $2$, while type-2 jobs are exponential with mean $1/2$.
The parameter $\lambda$ is set such that a predefined load $\rho$ is achieved. Note that
the instantaneous load when the arrival process is in state $2$ is higher than in state $1$,
while the mean time spent in state $i$ is $q_i$ for $i=1,2$.

\begin{figure*}[t!]
\begin{subfigure}[t]{0.48\textwidth}
  \centering
  \includegraphics[width=\linewidth]{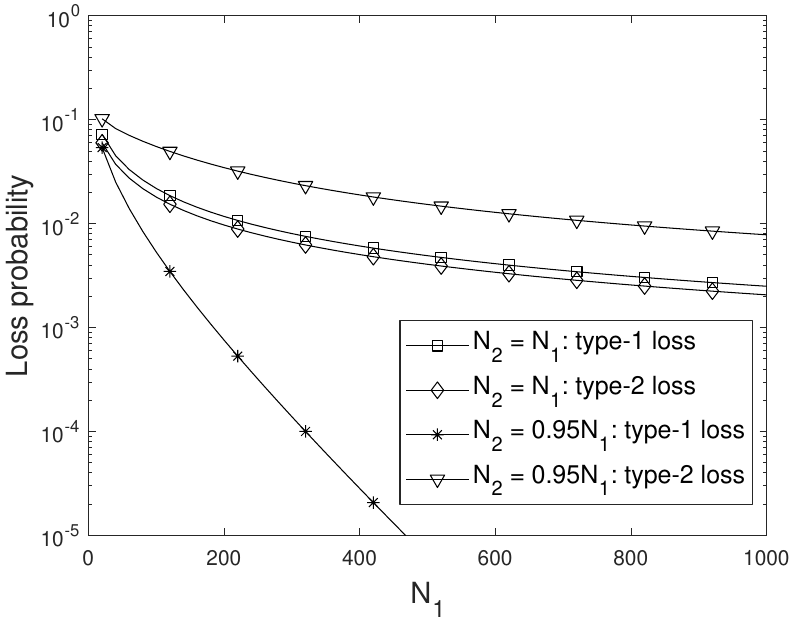}
  \caption{$\rho=1$}\label{fig:1a}
\end{subfigure}
\begin{subfigure}[t]{0.48\textwidth}
  \centering
  \includegraphics[width=\linewidth]{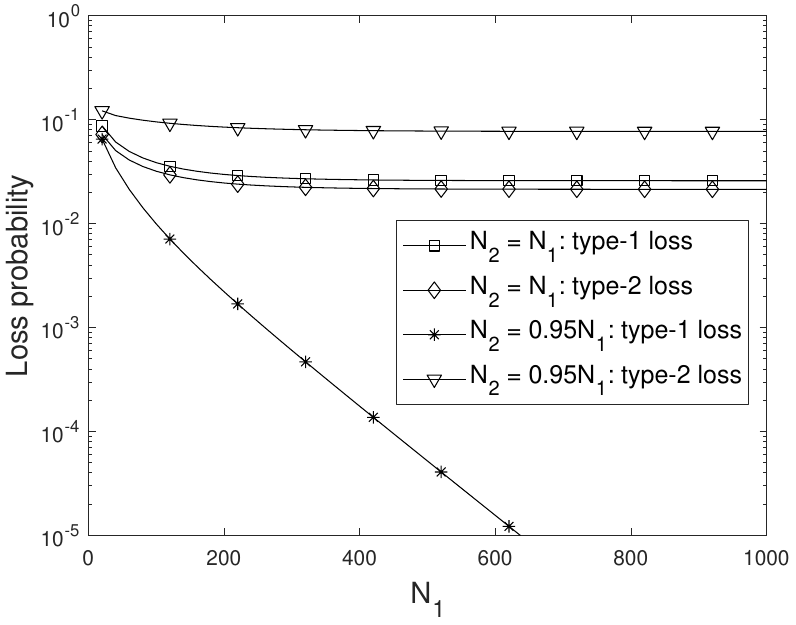}
  \caption{$\rho=1.025$}\label{fig:1b}
\end{subfigure}
\caption{Loss probability of an MMAP[L]/PH[L]/1/N[L]/LCFS queue with $L=2$ job types, exponential job durations
and $2$-state MMAP[L] arrivals. }
\label{fig:1}
\end{figure*}

In Figure \ref{fig:1} we plot the loss probability as a function of $N_1$ for a load $\rho=1$ and $\rho=1.025$
for both job types. Note that the Markov process is positive (Harris) recurrent for any load $\rho$ as the
queue does not allow any additional arrivals when it contains $N_1$ jobs.  For $N_2$ two cases are considered:
either $N_2=N_1$, in which case the loss probability of both types is similar, or
$N_2 = 0.95N_1$, which means the final part of the queue cannot be used by type-$2$ arrivals and
therefore the loss probability decays much faster for the type-$1$ jobs (at the expense of a
slower decay for the type-$2$ jobs).

Figure \ref{fig:2} presents some results for the queue length distribution using the same MMAP[L] arrival
process and the same service times. In this figure $N_1=1000$ and five settings
for $N_2$ are considered: $0, 250, 500, 750$ and $1000$. When $\rho=1$ the queue length distribution is more or less uniform
between $0$ and $N_2$, except for values close to $0$ and $N_2$. When $\rho > 1$ the probability of having $q$
jobs in the queue initially increases with $q$ and reaches a maximum in $q=N_2$. The queue length distribution starts to
decay quickly as soon as $q$ exceeds $N_2$ as the load of the type-1 jobs only is $0.7$.

\begin{figure*}[t!]
\begin{subfigure}[t]{0.48\textwidth}
  \centering
  \includegraphics[width=\linewidth]{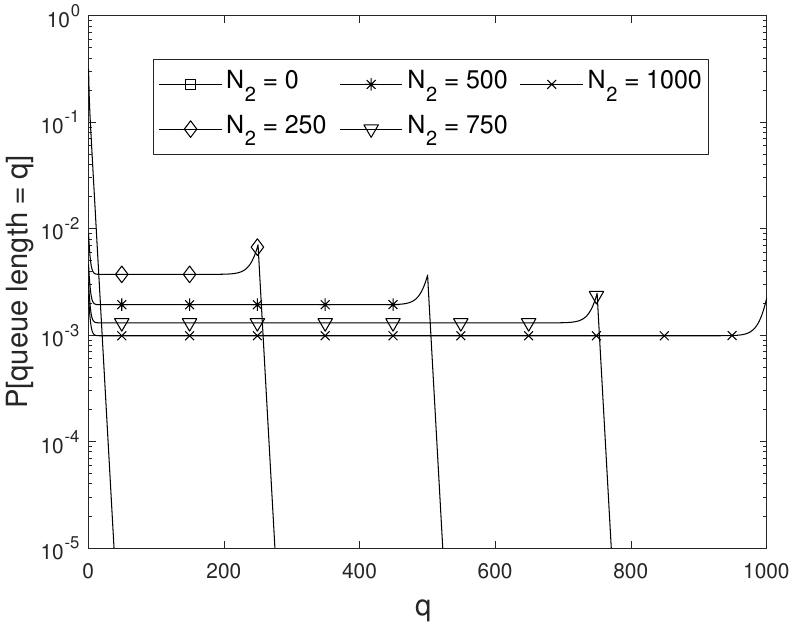}
  \caption{$\rho=1$}\label{fig:2a}
\end{subfigure}
\begin{subfigure}[t]{0.48\textwidth}
  \centering
  \includegraphics[width=\linewidth]{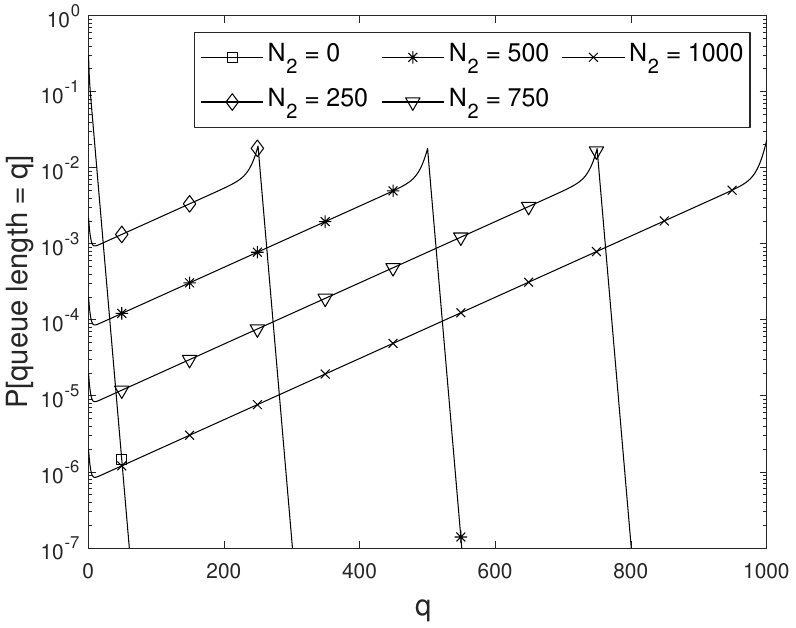}
  \caption{$\rho=1.025$}\label{fig:2b}
\end{subfigure}
\caption{Queue length distribution of an MMAP[L]/PH[L]/1/N[L]/LCFS queue with $L=2$ job types, exponential job durations
and $2$-state MMAP[L] arrivals. The queue can hold up to $1000$ jobs (that is, $N_1 = 1000$).}
\label{fig:2}
\end{figure*}

\subsection{Runtime results for the finite FCFS queue with a "multi-level job" cascade}
We consider the queue introduced in Section \ref{sec:apps3}. It is not hard to see that if the number of 
job levels equals $C$ and the phase-type representation of level $c$ jobs equals $k$ for all $c$, then
the phase-type representation of a single "multi-level" job equals $\sum_{c=1}^C k^c = (k^{C+1}-1)/(k-1)-1$,
as $k^c$ phases are needed when a level $c$ job is being executed (for all $c$).
This implies that the classic Quasi-Birth-Death Markov chain used to solve the MAP/PH/1/N queue has blocks
that grow exponential in $C$, meaning the runtime scales as $k^{3C}$. For the colored MMFQ, the size of the
matrices involved does not depend on $C$. In this case the number of matrices involved grows linear in $C$ and therefore
so does the runtime.

\begin{figure*}[t!]
  \centering
  \includegraphics[width=0.5\linewidth]{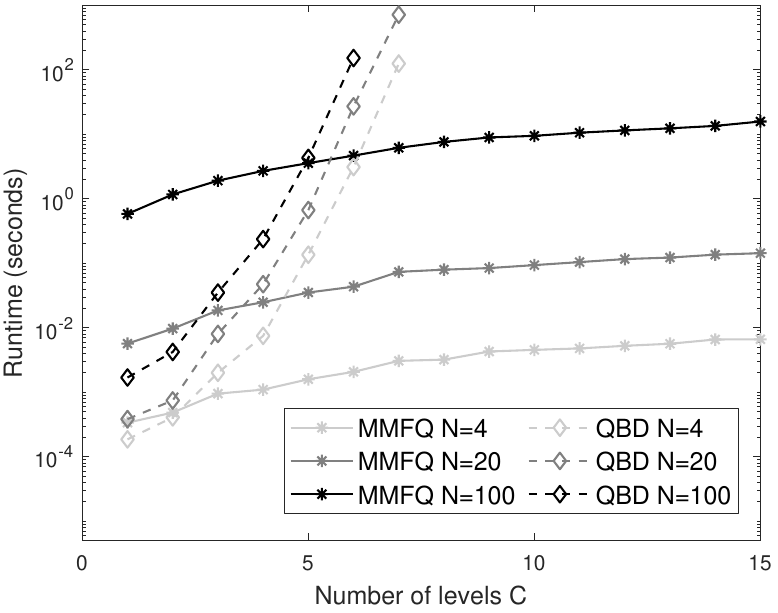}
  \caption{Runtime of the colored MMFQ framework versus the Quasi-Birth-Death Markov chain to compute the
  queue length distribution of the finite FCFS queue with a "multi-level job" cascade}\label{fig:runtime}
\end{figure*}

Figure \ref{fig:runtime} presents the runtime results for a MATLAB implementation on an Intel(R) Core(TM) i7-8550U CPU @ 1.80GHz processor with 16Gb memory. The timing was performed  using the {\it timeit} MATLAB function.
It shows that the runtime indeed increases linearly in $C$ for the colored MMFQ framework, while it increases
exponentially in $C$ for the Quasi-Birth-Death approach. When $C=7$ each of the block matrices of the Quasi-Birth-Death
Markov chain requires approximately $344$MB of memory for a $2$-state MAP. For $C=8$ the memory requirement for each
block matrix would be around 3GB of memory, making the computation infeasible on a 16GB machine. 

The precise parameter setting for the runtime results presented in Figure \ref{fig:runtime} are as folllows.
The MAP is a $2$-state interrupted Poisson process with a mean sojourn time in both states of $100$. In state $i$
the MAP generates 
arrivals at rate $\lambda_i$, where $\lambda_1=0$ and $\lambda_2$ is set such that the load of the queue equals
$0.8$. For the "multi-level" jobs we used an Erlang-$3$ distributions at level $c$ with mean $1/\mu_c = 1/(1.1)^c$
for $c=1,\ldots,C$.
The rate at which level $c+1$ jobs are spawned $\gamma_{c+1}$ was set equal to $0.9^c$ for $c=1,\ldots,C-1$. 

\section{Proof of Theorem \ref{th:main}}\label{sec:proof}

To ease the presentation we first prove Theorem \ref{th:main} for $C=2$ colors, 
the generalization to $C > 2$ colors is not hard and is briefly discussed in Appendix \ref{app:genC}. 
We start by presenting a set of partial differential equations (PDEs) that captures
the evolution of the colored MMFQ. This theorem is then used to prove Theorem \ref{th:main}.

\subsection{A set of PDEs for a $2$-colored MMFQ}

Let $f_i(x,y;t)$ be the probability that at time $t$ there is $x$ fluid of color one and $y$ fluid of color two
and the background state is $i$. Recall that either
\begin{itemize}
    \item $i \in S_-$,
    \item  $i \in S_+^{(2)}$ and $y > 0$, or
    \item $i \in S_+^{(1)}$, $x > 0$ and $y=0$. 
\end{itemize}
Let $f_-(x,y;t)$ be the vector holding the $f_i(x,y;t)$ values for $i \in S_-$,
$f_+(x,y;t)$ the vector holding the $f_i(x,y;t)$ values for $i \in S_+^{(2)}$ when $y > 0$ and
$f_+(x,0;t)$ the vector holding the $f_i(x,0;t)$ values for $i \in S_+^{(1)}$. 

\begin{theorem}\label{th:PDEs}
    The $f_i(x,y;t)$ values satisfy the following set of partial differential equations:
   \begin{align}
       \frac{\partial f_-(x,y;t)}{\partial t}  - \frac{\partial f_-(x,y;t)}{\partial y}  &=
       f_-(x,y;t) T_{--}^{(2)}  + f_+(x,y;t) T_{+-}^{(2)}, \label{eq:y>0_-}\\
        \frac{\partial  f_+(x,y;t)}{\partial t} + \frac{\partial f_+(x,y;t)}{\partial y}  &=
       f_-(x,y;t) T_{-+}^{(2)}  + f_+(x,y;t) T_{++}^{(2)}, \label{eq:y>0_+}
      \end{align} 
   for $x \geq 0, y > 0$ and
   \begin{align}
       \frac{\partial f_-(x,0;t)}{\partial t}  - \frac{\partial f_-(x,0;t)}{\partial x}  &=
       f_-(x,0;t) T_{--}^{(1)}  + f_+(x,0;t) T_{+-}^{(1)} + f_-(x^+,0^+;t), \label{eq:x>0_-}\\
        \frac{\partial  f_+(x,0;t)}{\partial t} + \frac{\partial f_+(x,0;t)}{\partial x}  &=
       f_-(x,0;t) T_{-+}^{(1)}  + f_+(x,0;t) T_{++}^{(1)}, \label{eq:x>0_+}
      \end{align} 
  for $x > 0$, where $f_-(x^+,0^+,t) = \lim_{y \rightarrow 0^+} f_-(x+y,y;t)$.
  For $x=y=0$, we have
   \begin{align}\label{eq:xy00}
       \frac{\partial f_-(0,0;t)}{\partial t} &=
       f_-(0,0;t) T_{--}^{(0)}  +f_-(0^+,0;t) + f_-(0,0^+;t), 
     \end{align} 
  with $f_-(0^+,0;t) = \lim_{x \rightarrow 0^+} f_-(x,0;t)$ and $f_-(0,0^+;t) = \lim_{y \rightarrow 0^+} f_-(0,y;t)$.
  Finally, we have the boundary conditions:
  \begin{align}\label{eq:boundary2}
       f_+(x,0^+;t) &=  f_-(x,0;t) T_{-+}^{(1,2)}  +  f_+(x,0;t) T_{++}^{(1,2)}, \\
    f_+(0,0^+;t) &= f_-(0,0,t) T_{-+}^{(0,2)}, \label{eq:boundary3}\\
     f_+(0^+,0;t) &= f_-(0,0,t) T_{-+}^{(0,1)}, \label{eq:boundary4}
     \end{align} 
     where $f_+(x,0^+;t) = \lim_{y \rightarrow 0^+} f_+(x,y;t)$.
\end{theorem}
\begin{proof}
The proof is presented in Appendix \ref{app:theo2}.
\end{proof}
\stepcounter{remarkcount}
\paragraph{Remark \arabic{remarkcount}:} If we set the derivative with respect to $t$ in \eqref{eq:y>0_-}-\eqref{eq:boundary4} equal to zero, we obtain
a set of PDEs that we refer to as {\bf the stationary set of PDEs}. This stationary set has at most one partial derivative in each equation
and all equations are linear. Further, the stationary set of PDEs has a unique solution given the vectors $p_-=\pi_-(0,0)$ and $\pi_-(x,0^+)=\pi_-(x^+,0^+)$ for any $x \geq 0$.  This can be noted
as follows. The vector $p_-$ uniquely determines $\pi_+(0^+,0)$ via  \eqref{eq:boundary4}.
The vector $\pi_-(0^+,0)$ is obtained from $p_-$ and $\pi_-(0,0^+)$ using \eqref{eq:xy00}. Having obtained
$\pi_+(0^+,0)$ and $\pi_-(0^+,0)$ and using the given values for $\pi_-(x^+,0^+)$, we get a unique solution for
$\pi_+(x,0)$ and $\pi_-(x,0)$ for any $x>0$ via \eqref{eq:x>0_-} and \eqref{eq:x>0_+}, which corresponds to a
nonhomogeneous linear set of ODEs. More specifically, we have
\[[\pi_+(x,0),\pi_-(x,0)] = [\pi_+(0^+,0),\pi_-(0^+,0)] e^{T^{(1)}V^{(1)}x}-\int_{0}^x [0,\pi_-(x^+,0^+)] e^{T^{(1)}V^{(1)}(x-z)}dz,\]
where $V^{(c)}$ is a diagonal matrix holding $|S_+^{(c)}|$ ones, followed by $|S_-|$ minus ones, for $c=1,2$.
From these, \eqref{eq:boundary2} yields
$\pi_+(x,0^+)$ for any $x > 0$. The value for $\pi_+(0,0^+)$ is found using \eqref{eq:boundary3} as $p_-$ is given.
Hence, we know all $\pi_+(x,0^+)$ for $x\geq 0$ and have all the $\pi_-(x,0^+)$ values for $x\geq 0$ as given. 
$\pi_+(x,y)$ and $\pi_-(x,y)$ for $x \geq 0$ and $y>0$ can then be expressed as using \eqref{eq:y>0_-} and \eqref{eq:y>0_+}, 
which correspond to a homogeneous linear set of ODEs, as such
\[[\pi_+(x,y),\pi_-(x,y)] = [\pi_+(x,0^+),\pi_-(x,0^+)] e^{T^{(2)}V^{(2)}x},\]
for $x \geq 0$ and $y > 0$.

\refstepcounter{remarkcount}
\paragraph{Remark \arabic{remarkcount}\label{rem:A}:} The stationary set of PDEs also has a unique solution given $p_-$ and the requirement
that $\pi_-(x^+,0^+)=\pi_-(x,0^+)=\pi_+(x,0^+)A$ for some matrix $A$. As before $p_-$ uniquely determines $\pi_+(0^+,0)$ via  \eqref{eq:boundary4} and now $\pi_-(0^+,0) = \pi_+(0^+,0)A$. Combining \eqref{eq:x>0_-}, \eqref{eq:x>0_+} and \eqref{eq:boundary2}
with the fact that $\pi_-(x^+,0^+)=\pi_+(x,0^+)A$, yields a homogeneous linear set of ODEs for $[\pi_+(x,0),\pi_-(x,0)]$ given by
   \begin{align*}
       - \frac{\partial \pi_-(x,0)}{\partial x}  &=
       \pi_-(x,0) (T_{--}^{(1)}+T_{-+}^{(1,2)}A)  + \pi_+(x,0) (T_{+-}^{(1)}+T_{++}^{(1,2)}A) \\
       \frac{\partial \pi_+(x,0)}{\partial x}  &=
       \pi_-(x,0) T_{-+}^{(1)}  + \pi_+(x,0) T_{++}^{(1)},
      \end{align*} 
which has the unique solution
\begin{align*}
 [\pi_+(x,0),\pi_-(x,0)] = [\pi_+(0^+,0),\pi_-(0^+,0)] e^{T_AV^{(1)}x},      
      \end{align*} 
      with $T_A = \begin{bmatrix}
     T_{++} & T_{+-}^{(1)}+T_{++}^{(1,2)}A \\ T_{-+} & T_{--}^{(1)}+T_{-+}^{(1,2)}A
 \end{bmatrix}$ for $x>0$. By \eqref{eq:boundary2} we therefore also have $\pi_+(x,0^+)$ for $x > 0$.
 $\pi_+(0,0^+)$ is given by $\eqref{eq:boundary3}$ and  $\pi_-(x,0^+) = \pi_+(x,0^+)A$ for $x \geq 0$.
 The values for $[\pi_+(x,y),\pi_-(x,y)]$ again follow from \eqref{eq:y>0_-} and \eqref{eq:y>0_+} for $y > 0$
 by solving a homogeneous linear set of ODEs.

\subsection{Proof of Theorem \ref{th:main} for $C=2$ colors}

Let $\pi_i(x,y)=\lim_{t\rightarrow \infty} f_i(x,y;t)$ be the steady state densities and
$(p_-)_i = \lim_{t \rightarrow \infty} f_i(0,0;t)$ be the probability mass of the boundary states, then $\pi_i(x,y)$ and $(p_-)_i$ satisfy the stationary version of
\eqref{eq:y>0_-}-\eqref{eq:boundary4}, where
the partial derivatives with respect to $t$ are replaced by zero.

    The proof of Theorem \ref{th:main} consists in showing that for the vector $p_-$ that corresponds to the stationary probability
    vector of the colored MMFQ censored on the states $\{(0,i)|i\in S_-\}$, the stationary set of PDEs has a 
    solution given by  (with $x,y>0$):
    \begin{align}
    [\pi_+(x,y),\pi_-(x,y)] &= p_- T_{-+}^{(0,1)} e^{K_1x}(T_{++}^{(1,2)}+\Psi_1 T_{-+}^{(1,2)}) e^{K_2y}[I,\Psi_2], \\
    [\pi_+(x,0),\pi_-(x,0)] &= p_- T_{-+}^{(0,1)} e^{K_1x}[I,\Psi_1], \\
    [\pi_+(0,y),\pi_-(0,y)] &= p_- T_{-+}^{(0,2)} e^{K_2y}[I,\Psi_2]. 
    \end{align}
    Note that for this solution we have $\pi_-(x^+,0^+)=\pi_-(x,0^+)=\pi_+(x^+,0^+)\Psi_2$, so by Remark \ref{rem:A}
    it is the  unique solution of the stationary set of PDEs for this $p_-$.
    
    Consider the stationary version of \eqref{eq:y>0_+}, which yields that
    \[ p_- T_{-+}^{(0,1)}  e^{K_1x}(T_{++}^{(1,2)}+\Psi_1 T_{-+}^{(1,2)}) e^{K_2y}K_2 = p_- T_{-+}^{(0,1)}  e^{K_1x}(T_{++}^{(1,2)}+\Psi_1 T_{-+}^{(1,2)}) e^{K_2y} [\Psi_2 T_{-+}^{(2)}+T_{++}^{(2)}], \]
    must hold for $x > 0$ and
    \[ p_- T_{-+}^{(0,2)}  e^{K_2y} K_2 = p_- T_{-+}^{(0,2)}  e^{K_2y} [\Psi_2 T_{-+}^{(2)}+T_{++}^{(2)}], \]
    for $x=0$.     These two equalities are clearly satisfied as 
    \begin{align}\label{eq:Cmat}
        K_2 = \Psi_2 T_{-+}^{(2)}+T_{++}^{(2)},
    \end{align}
    by definition of $K_2$.
    
 Proceed with the stationary version of \eqref{eq:y>0_-}  to find that
       \[ - p_- T_{-+}^{(0,1)}  e^{K_1x}(T_{++}^{(1,2)}+\Psi_1 T_{-+}^{(1,2)}) e^{K_2y} K_2 \Psi_2 = p_- T_{-+}^{(0,1)}  e^{K_1x}(T_{++}^{(1,2)}+\Psi_1 T_{-+}^{(1,2)}) e^{K_2y} [\Psi_2 T_{--}^{(2)} + T_{+-}^{(2)}],\]
       for $x>0$ and
       \[ - p_- T_{-+}^{(0,2)} e^{K_2y} K_2 \Psi_2 = p_- T_{-+}^{(0,2)}  e^{K_2y} [\Psi_2 T_{--}^{(2)} + T_{+-}^{(2)}],\]
       for $x=0$. Therefore, it suffices that
    \begin{align}\label{eq:Dmat}
        0 = K_2\Psi_2+ \Psi_2 T_{--}^{(2)}+T_{+-}^{(2)} = \Psi_2 T_{-+}^{(2)}\Psi_2+T_{++}^{(2)}\Psi_2+ \Psi_2 T_{--}^{(2)}+T_{+-}^{(2)}, 
    \end{align}
       due to \eqref{eq:Cmat}. This equation is satisfied by definition of $\Psi_2$, see \eqref{eq:NAREC}.
    
    The stationary version of \eqref{eq:boundary2} demands that
\[  p_- T_{-+}^{(0,1)}  e^{K_1x}  (T_{++}^{(1,2)}+\Psi_1 T_{-+}^{(1,2)}) = p_- T_{-+}^{(0,1)}  e^{K_1x}\Psi_1 T_{-+}^{(1,2)} +
 p_- T_{-+}^{(0,1)}  e^{K_1x}T_{++}^{(1,2)}, \]
which clearly holds.
Focusing on the stationary version of \eqref{eq:x>0_+} we observe that
\[  p_- T_{-+}^{(0,1)}  e^{K_1x} K_1 = p_- T_{-+}^{(0,1)}  e^{K_1x}[\Psi_1 T_{-+}^{(1)} +T_{++}^{(1)}], \]
is required, which holds by definition of $K_1$.

From \eqref{eq:x>0_-} we get the requirement that
\[  -p_- T_{-+}^{(0,1)}  e^{K_1x} K_1 \Psi_1 = p_- T_{-+}^{(0,1)}  e^{K_1x}[\Psi_1 T_{--}^{(1)} +T_{+-}^{(1)}+ (T_{++}^{(1,2)}+\Psi_1 T_{-+}^{(1,2)}) \Psi_2]. \]
This equality is satisfied when
\begin{align*}
    0 &= K_1\Psi_1 +\Psi_1 T_{--}^{(1)} +T_{+-}^{(1)}+(T_{++}^{(1,2)}+\Psi_1 T_{-+}^{(1,2)})\Psi_2 \\
    &= \Psi_1 T_{-+}^{(1)}\Psi_1 +T_{++}^{(1)}\Psi_1+
\Psi_1 T_{--}^{(1)} +T_{+-}^{(1)}+ \Psi_1 T_{-+}^{(1,2)}\Psi_2 +T_{++}^{(1,2)}\Psi_2,
\end{align*}
which is equivalent to \eqref{eq:NAREs} and therefore holds. 
The boundary conditions \eqref{eq:boundary3} and \eqref{eq:boundary4} hold
as $\pi_+(0,0^+) = p_- T_{-+}^{(0,2)}$ and $\pi_+(0^+,0) = p_- T_{-+}^{(0,1)}$, respectively. 
The requirement for $p_-$ now follows from the stationary version of 
\eqref{eq:xy00}:
    \begin{align*}
   0 &= p_- T_{--}^{(0)} + p_- T_{-+}^{(0,1)} \Psi_1  + p_- T_{-+}^{(0,2)}\Psi_2.    
    \end{align*}
    which corresponds to \eqref{eq:pmin}.
Note that $T_{--}^{(0)} +T_{-+}^{(0,1)} \Psi_1  +T_{-+}^{(0,2)}\Psi_2$ is the rate matrix of the MMFQ observed only when
the fluid equals zero (i.e., censored on the states $\{(0,i)|i\in S_-\}$), as such, the vector $p_-$ is the correct boundary vector and due to the stochastic
interpretation of the matrices $\Psi_1$ and $\Psi_2$, the condition
$\pi_-(x,0^+) = \pi_+(x,0^+)\Psi_2$ must hold. This implies that the distribution given in Theorem \ref{th:main} 
is the stationary distribution of the colored MMFQ, where the normalizing condition follows from
    \begin{align*}
1&=         p_- e +  \int_0^\infty (\pi_-(x,0)e + \pi_+(x,0)e) dx  + \int_0^\infty (\pi_+(0,y)e+\pi_-(0,y)e) dy \\
  &\hspace{1cm }+ \int_0^\infty\int_0^\infty (\pi_+(x,y)e+\pi_-(x,y)e) dx dy,
    \end{align*}
    as we need to obtain a distribution.

\section{Conclusions and future work}\label{sec:future}

In this paper we generalized the Markov-modulated fluid queue framework by introducing colored MMFQs and
colored MMFQs with fluid jumps. We developed a matrix analytic method to compute the stationary distribution of a  colored
MMFQ (with fluid jumps) and demonstrated that the new framework enables the analysis of queueing systems that would otherwise be
intractable due to a state-space explosion.
The novel framework can be further extended in many directions. Possible lines of future work include multi-layered 
colored MMFQs and colored MMFQs where
the colors are unordered. 

\bibliographystyle{abbrv}
\bibliography{thesis}

\appendix

\section{The matrices $T_{--}^{(c)}$ and $T_{-+}^{(c,C)}$ do not depend on $c$ and $T_{-+}^{(c)}=T_{-+}^{(c,\ell)}=0$ for $\ell < C$}
\label{app:special2}
Introducing colors in an MMFQ can be regarded as a way to add some form of memory to the MMFQ, where we
remember the fluid level by marking it with a color change. When the matrices $T_{--}^{(c)}$ or $T_{-+}^{(c,C)}$ 
 depend on $c$, the colored MMFQ makes use of this memory. 
 Similarly, as soon as one of the $T_{-+}^{(c,\ell)}$  or $T_{-+}^{(c)}$ matrices is nonzero for $c,\ell < C$ the memory is used.
 This can be understood by noting that the color cannot become $\ell < C$ when color $C$ is on top of the fluid, so
 the color cannot become $\ell$ from any other color if the color on top of the fluid does not influence the colored MMFQ
 when the fluid decreases. 
In this subsection we show that if the $T_{--}^{(c)}$ and $T_{-+}^{(c,C)}$ 
matrices are independent of $c$, meaning
\[T_{--}^{(c)} = T_{--}, \ \ \   \mbox{ and } \ \ \ 
T_{-+}^{(c,C)} = T_{-+}^{(C)},\] 
for all $c$  and $T_{-+}^{(c)}=T_{-+}^{(c,\ell)}=0$ for $\ell < C$,  
 then the colored MMFQ 
essentially reduces to a classic MMFQ with a special structure.

To see this, define the matrix $\hat \Psi$ as 
\begin{align*}
    \hat \Psi =  \begin{bmatrix}
    \Psi_1 \\ \Psi_2 \\ \vdots \\  \Psi_C 
\end{bmatrix}. 
\end{align*}
When  $T_{--}^{(c)}$ 
and $T_{-+}^{(c,C)}$ do not depend on $c$, and $T_{-+}^{(c)}=T_{-+}^{(c,\ell)}=0$ for $\ell < C$,
the set of equations given by \eqref{eq:NAREC} and \eqref{eq:NAREs} becomes
\begin{align*}
    T_{++}^{(C)}\Psi_C + & T_{+-}^{(C)}  +\Psi_CT_{-+}^{(C)}\Psi_C +
   \Psi_C T_{--}= 0,\\
    T_{++}^{(c)}\Psi_c + & \left( T_{+-}^{(c)} + \sum_{\ell > c} T_{++}^{(c,\ell)}\Psi_\ell \right)
     +
   \Psi_c \left(T_{--}+ T_{-+}^{(C)}\Psi_C \right) = 0,
\end{align*}
for $c=1,\ldots, C-1$ .
We can combine these $C$ matrix equations in matrix form as
\begin{align*}
     \underbrace{\begin{bmatrix}
         T_{+-}^{(1)} \\
         \vdots\\
         T_{+-}^{(C)}
     \end{bmatrix}}_{\hat T_{+-}} + 
     \underbrace{\begin{bmatrix}
    T_{++}^{(1)} & T_{++}^{(1,2)} & \ldots & T_{++}^{(1,C)} \\
    0 & T_{++}^{(2)} & \ldots & T_{++}^{(2,C)} \\
    \vdots & \vdots & \ddots & \vdots \\
    0 & 0 & \ldots & T_{++}^{(C)} 
\end{bmatrix}}_{\hat T_{++}} \hat  \Psi
    +\hat \Psi \underbrace{T_{--}}_{\hat T_{--}}  +
   \hat \Psi \underbrace{\begin{bmatrix}
       0 & \ldots & 0 & T_{-+}^{(C)}
   \end{bmatrix}}_{\hat T_{-+}}\hat  \Psi = 0,
\end{align*}
which is the NARE of a classic MMFQ characterized by the matrices $\hat T_{--}, \hat T_{-+}, \hat T_{+-}$ and $\hat T_{++}$. The rate matrix $K$ given by \eqref{eq:K} becomes
\begin{align*}
    K = \hat T_{++} + \hat \Psi \hat T_{-+},
\end{align*}
as in the classic MMFQ. The same holds for the boundary vector and normalization condition.

\section{Proof Sketch of Theorem \ref{th:main} for general $C$}\label{app:genC}

The proof proceeds in the same fashion as the proof in Section \ref{sec:proof} for $C=2$ colors and makes use of the following theorem
that can be proven in a similar way as Theorem \ref{th:PDEs}.

Let $\Theta_c = \{(x_1,\ldots,x_c,0,\ldots,0) \in \mathbb{R}^C| x_c > 0\}$.
Let $f_i(\vec x;t)$ with $\vec x \in \Theta_c$ represent the state where there is $x_k$ fluid of color $k$
at time $t$ and the background state is $i$. Note that $c$ is the color that is currently
added or removed from the colored MMFQ queue. Then, either
\begin{itemize}
   \item $i \in S_-$, or
    \item  $i \in S_+^{(c)}$. 
\end{itemize}
Let $f_-(\vec x;t)$ be the vector holding the $f_i(\vec x;t)$ values for $\vec x \in \Theta_c$ and $i \in S_-$. Let
$f_+^{(c)}(\vec x;t)$ be the vector holding the $f_i(\vec x;t)$ values for $i \in S_+^{(c)}$.

Denote $\vec e_k$ as the $k$-th row of the size $C$ unity matrix.
When $\vec x \in \Theta_c$ and $k > c$, then  
$\vec y = \vec x+ \delta \vec e_k \in \Theta_{k}$ with $y_{c+1}=\ldots=y_{k-1}=0$.  

\begin{theorem}
    The $f_i(\vec x;t)$ values with $\vec x \in \Theta_c$ satisfy the following set of partial differential equations:
   \begin{align}
       \frac{\partial f_-(\vec x;t)}{\partial t}  - \frac{\partial f_-(\vec x;t)}{\partial x_{c}}  &=
       f_-(\vec x;t) T_{--}^{(c)}  + f_+^{(c)}(\vec x;t) T_{+-}^{(c)} +
         \sum_{k = c+1}^C f_-(\vec x^{(k+)};t), \label{eq:genx>0_-}\\
        \frac{\partial  f_+^{(c)}(\vec x;t)}{\partial t} + \frac{\partial f_+^{(c)}(\vec x;t)}{\partial x_{c}}  &= f_-(\vec x;t) T_{-+}^{(c)}  + f_+^{(c)}(\vec x;t) T_{++}^{(c)}, \label{eq:geny>0_+}
      \end{align} 
    where $f_-(\vec x^{(k+)},t) = \lim_{\delta \rightarrow 0^+} f_-(\vec x+ \vec e_k \delta;t)$.
  Further,
   \begin{align}\label{eq:genxy00}
       \frac{\partial f_-(\vec 0;t)}{\partial t} &=
       f_-(\vec 0;t) T_{--}^{(0)}  + \sum_{k = 1}^C f_-(\vec 0^{(k+)};t). 
     \end{align} 
  In addition, the following boundary conditions hold:
  \begin{align}\label{eq:boundaryC}
       f_+(\vec x^{(k+)};t) &=  f_-(\vec x;t) T_{-+}^{(c,k)}  +  f_+(\vec x;t) T_{++}^{(c,k)}, \\
    f_+(\vec 0^{(k+)};t) &= f_-(\vec 0,t) T_{-+}^{(0,k)}, \label{eq:boundaryk}
     \end{align} 
     for $k=1,\ldots,C$, $k>c$ and $\vec x \in \Theta_c$.
\end{theorem}

\section{Proof of Theorem \ref{th:PDEs}}\label{app:theo2}
    We start with \eqref{eq:y>0_-}, that is, the case where $y > 0$ and $i \in S_-$. 
    Let $0 < \Delta < y$, suppose we want to be in state $(x,y,i)$ at time $t + \Delta$.
    There are three ways to be in state $(x,y,i)$ at time $t+\Delta$ that are not $o(\Delta)$:
    (i) the fluid simply decreases in $(t,t+\Delta)$, (ii) there is a phase change of the background process
    from another phase $j \in S_-$ in $(t,t+\Delta)$ and (iii) there is a phase change from a 
    phase $j\in S_+^{(2)}$ in $(t,t+\Delta)$.
    For the first two events the fluid must equal $x+\Delta$ at time $t$. The
    first event takes place with probability 
    \[f_i(x,y+\Delta;t) (1+(T_{--})_{i,i}\Delta) + o(\Delta).\] 
    The second event occurs with
    probability $ f_j(x,y+\Delta;t)(T_{--})_{j,i} \Delta$. For the last event the
    fluid increases for some time $v \in (0,\Delta)$ and then decreases, so the fluid at time
    $t$ must equal $y-v+(\Delta-v)$. Hence, the last event happens with probability
    \[ \int_0^\Delta f_j(x,y+\Delta -2v;t) (T_{+-})_{j,i} dv.\]
    Combining this, dividing by $\Delta$ and taking limits, we have
    \begin{align*}
        \lim_{\Delta \rightarrow 0^+} &\frac{f_i(x,y;t+\Delta)-f_i(x,y;t)+f_i(x,y;t) - f_i(x,y+\Delta;t)}{\Delta} = \\
        &f_i(x,y;t) (T_{--})_{i,i} + 
        \lim_{\Delta \rightarrow 0^+} \int_0^\Delta \frac{f_j(x,y+\Delta -2v;t)}{\Delta} dv  \ (T_{+-})_{j,i}. 
     \end{align*}
     Therefore \eqref{eq:y>0_-} follows as 
     \begin{align*}
         \lim_{\Delta \rightarrow 0^+} & \int_0^\Delta \frac{f_j(x,y+\Delta -2v;t)}{\Delta} dv = 
         -\frac{1}{2}\lim_{\Delta \rightarrow 0^+} \int_{y+\Delta}^{y-\Delta} \frac{f_j(x,u;t)}{\Delta} du =\\
         &\frac{1}{2}\lim_{\Delta \rightarrow 0^+} \frac{F_j(x,y+\Delta)-F_j(x,y)+F_j(x,y)-F_j(x,y-\Delta)}{\Delta} = 
         f_j(x,y;t),
     \end{align*}
     where $\frac{\partial F_j(x,y;t)}{\partial y} = f_j(x,y;t)$.
     For \eqref{eq:y>0_+} an analogous argument can be used: the fluid goes up in $(t,t+\Delta)$ 
     in case of the first two events (no phase change or a change from a positive to a positive phase), while
     for the third event it goes down for a time $v$ and then up for a time $\Delta-v$. This yields for $y > 0$
     and $i \in S_+^{(2)}$
    \begin{align*}
        \lim_{\Delta \rightarrow 0^+} &\frac{f_i(x,y;t+\Delta)-f_i(x,y;t)+f_i(x,y;t) - f_i(x,y-\Delta;t)}{\Delta} = \\
        &f_i(x,y;t) (T_{++})_{i,i} + 
        \lim_{\Delta \rightarrow 0^+} \int_0^\Delta \frac{f_j(x,y-\Delta +2v;t)}{\Delta} dv  \ (T_{-+})_{j,i}. 
     \end{align*}
     This implies \eqref{eq:y>0_+}. Note that \eqref{eq:y>0_-} and \eqref{eq:y>0_+} are also valid for $x=0$.
     
     Proceeding with \eqref{eq:x>0_-}, we have three similar events as for \eqref{eq:y>0_-} to reach
     state $(x,0,i)$ at time $t+\Delta$: (i) there is no phase change, (ii) there is a phase changes between
     two states in $S_-$ and (iii) there is a phase change from $s_+^{(1)}$ to $S^-$. These yield the same terms
     as in \eqref{eq:y>0_-} with the superscript $^{(2)}$ replaced by $^{(1)}$, with $y=0$ and a partial derivative
     with respect to $x$ instead of $y$. There is however a fourth event that can result in state $(x,0,i)$ at time $t+\Delta$.
     It could be that the state is $(x+\Delta-v,v,i)$ at time $t$ and there is no phase change. This event occurs
     with probability
     \[ \int_0^{\Delta} f_i(x+\Delta-v,v;t) dv   (1+(T_{--})_{i,i}\Delta + o(\Delta)). \]
     If we divide by $\Delta$, take the limit and set $u=\Delta-v$ we get
     \[ \lim_{\Delta \rightarrow 0^+} \frac{1}{\Delta} \int_0^{\Delta} f_i(x+u,\Delta-u;t) du. \]
    So we are taking the limit of the average of $f_i(x,y;t)$ over the diagonal line segment in $[x,x+\Delta]\times[0,\Delta]$,
    which corresponds to $\lim_{y \rightarrow 0^+}f_i(x+y,y;t)$. This results in the last term in \eqref{eq:x>0_-}.
    Note that all other events where $y$ becomes zero combined with a phase change are $o(\Delta)$ in probability.

    The reasoning for \eqref{eq:x>0_+} is similar to \eqref{eq:y>0_+}. We do not need to add an extra term as in
    \eqref{eq:x>0_-} because any event that would make $y$ zero requires a phase change
    when $i \in S_+^{(1)}$ and therefore these events are $o(\Delta)$.
    For \eqref{eq:xy00} we note that the boundary state $(0,0,i)$ with $i \in S_-$ 
    at time $t + \Delta$ can either be reached
    starting at time $t$ from state $(0,0,i)$ without a phase change or from some state $(0,0,j)$, for $j \not= i$, with a phase change. State $(0,0,i)$ can also be reached from states of the form $(x,0,i)$ or $(0,y,i)$ yielding the last two
    terms.

    In order to be in some state $(x,v,i)$ with $v \in (0,\Delta)$ and $i \in S_+^{(2)}$ at time $t+\Delta$ two events can occur that are
    not $o(\Delta)$: (a) there is a jump from some state $(x,0,j)$ with $j \in S_-$ in $(t,t+\Delta) $ or (b) there is a jump from some state $(x,0,j)$ with $j \in S_+^{(1)}$ in $(t,t+\Delta)$. Hence,
    \[ \int_0^\Delta f_+(x,v;t+\Delta)dv = 
    \int_0^\Delta f_-(x+v,0;t) T_{-+}^{(1,2)}dv + \int_0^\Delta f_+(x-v,0;t) T_{++}^{(1,2)}dv  + o(\Delta).\]
    Dividing by $\Delta$ and taking limits for $\Delta$ to zero yields the expression for $f_+(x,0^+;t)$ in \eqref{eq:boundary2}.

\end{document}